\documentclass[12pt,oneside]{amsart}
\usepackage{amsthm,amssymb}
\usepackage{geometry}                % See geometry.pdf to learn the layout options. There are lots.
\usepackage{mathrsfs}
\usepackage{mathtools}

\usepackage{xcolor}
\usepackage{hyperref}

\newcommand{\mc}[1]{\begin{bmatrix*}[c]#1\end{bmatrix*}}
\newcommand{\mr}[1]{\begin{bmatrix*}[r]#1\end{bmatrix*}}

\newtheorem{theorem}{Theorem}
\newtheorem{lemma}{Lemma}
\newtheorem{proposition}{Proposition}

\theoremstyle{definition}
\newtheorem{definition}{Definition}
\newtheorem{remark}{Remark}

\newcommand\res{\operatorname{res}}

\newcommand{\T}{\mathbf{T}}
\newcommand{\R}{\mathbf{R}}
\newcommand{\bT}{\mathbf{T}}
\newcommand{\bS}{\mathbf{S}}
\newcommand{\bH}{\mathbf{H}}
\newcommand{\bu}{\mathbf{u}}
\newcommand{\bPhi}{\mathbf{\Phi}}
\newcommand{\uu}{\mathfrak u}
\newcommand{\vv}{\mathfrak v}

\newcommand\Z{\mathbb{Z}}

\newcommand{\e}{{\rm e}}

\title{Integrability conditions for Boussinesq type systems}
\author{R. Hern\'andez Heredero}
\address{ETSIS de Telecomunicaci\'on, Universidad Polit\'ecnica de Madrid, Madrid, Spain}
\author{V. Sokolov}
\address{Higher School of Modern Mathematics MIPT,  Moscow,   Russia }
 
\date{\today\ %\currenttime
   }

\begin{document}
\begin{abstract}
The symmetry approach to the classification of evolution integrable partial differential equations (see, for example~\cite{MikShaSok91}) produces an infinite series of functions, defined in terms of the right hand side, that are conserved densities of any equation having infinitely many infinitesimal symmetries. For instance, the function $\frac{\partial f}{\partial u_{x}}$ has to be a conserved density of any integrable equation of the~KdV type~$u_t=u_{xxx}+f(u,u_x)$. This fact imposes very strong conditions on the form of the function~$f$. In this paper we construct similar canonical densities for equations of the Boussinesq type. In order to do that, we write the equations as evolution systems and generalise the formal diagonalisation procedure proposed in \cite{MSY} to these systems.
\end{abstract} 

\maketitle

\section{Introduction}

Integrable systems constitute an important class of equations because they have multiple applications in mathematical physics and pure mathematics. They are present in fluid mechanics, nonlinear optics, field theory, numerical analysis, algebraic geometry, etc., and explain phenomena such as the soliton and instantons, besides appearing in spectral theory and the theory of differential and pseudo-differential operators. Its study boomed since the discovery of the inverse scattering transform in the 1960'ies, and has not stopped since then.

The main feature of integrable systems is that they possess an infinite number of non-trivial symmetries and a bihamiltonian structure that, in nonlinear systems, relates the symmetries to an infinite number of non-trivial conservation laws. 

The \emph{symmetry approach} (cf.~\cite{MikShaSok91,Olv93}) provides a powerful, algorithmic method to study integrability, defining ways to derive important objects such as recursion, symplectic or Hamiltonian operators admitted by integrable systems. The symmetry approach has been extensively applied to scalar evolution equations, such as those of KdV type, systems of evolution equations like NLS type equations, or Boussinesq type systems of equations. As the complexity of the system under study grows, for example in the number of independent or dependent variables involved, or for non-evolutionary equations, the approach has to be further developed. In this paper we have systematised and extended the symmetry approach to the study of equations of the Boussinesq type, defining the concept of regularly diagonalisable systems and providing a general formula for canonical conserved densities of certain subclass of Boussinesq type of systems. We write explicitly the first canonical conservation laws of a subclass of fourth order Boussinesq type family of equations, that we finally classify.

Non-evolution  (1+1)-dimensional PDEs of the form
\begin{equation}\label{eq:eqnev}
u_{tt}=f(u,u_1,\ldots,u_{n};u_{t},u_{1t},\ldots,u_{mt}), \qquad u_i = \frac{\partial^{i}u}{\partial x^i}, \quad u_{it}=\frac{\partial^{i+1}u}{\partial x^i\partial t}
\end{equation}
are called {\it equations of Boussinesq type}.
Such equations can be rewritten as evolution systems of two variables as 
\begin{equation}\label{eq:sysgen}
\begin{aligned}
u_t&=v,\\
v_t&=f(u,u_1,\ldots,u_{n},v,v_1,\ldots,v_{m}).
\end{aligned}
\end{equation}
In this paper we consider the class of integrable\footnote{In our paper integrability means the existence of local infinitesimal symmetries and/or conservation laws (see, for example, \cite{zibshab1, Fokas, AbeGal83}).} systems \eqref{eq:sysgen} of the form
\begin{equation}\label{eq:eqsys}
\begin{aligned}
u_t&=v,\\
v_t&=f(u,u_1,\ldots,u_{2k},v,v_1,\ldots,v_{k-1}),
\end{aligned}  \qquad k>1.
\end{equation}
The Boussinesq equation $u_{tt}=u_{4}+ 2 u_1 u_{2}$, written as the system
\begin{align*}
 u_t&=v,\\
v_t&=u_4+ 2 u_1 u_2
\end{align*}
belongs to this class.

\begin{remark}\label{rem:sdx}
A similar class of integrable systems 
\begin{equation}\label{eq:eqsysMNW}
\begin{aligned}
u_t&=v_1,\\
v_t&=u_{2k-1}+g
\end{aligned}
\end{equation} 
was studied in~\cite{MNW}.
If $f_u=0$, then setting $u_1\to\tilde{u}$ allows to write~\eqref{eq:eqsys} as
\[\begin{aligned}
\tilde{u}_t&=v_1,\\
v_t&=f(\tilde{u},\ldots,\tilde{u}_{2k-1},v,v_1,\ldots,v_{k-1}).
\end{aligned}
\]
If $f=Dg$ (where~$Dg$ will denote the total $x$-derivative of a differential function~$g$ throughtout the paper), setting~$v\to\tilde{v}_1$ leads to
\[
\begin{aligned}
u_t&=\tilde{v}_1,\\
\tilde{v}_t&=g(u,u_1,\ldots,u_{2k-1},\tilde{v}_1,\ldots,\tilde{v}_k)
\end{aligned}
\]
so the corresponding function~$g$ on~\eqref{eq:eqsysMNW} must not depend on~$v$. Thus both families of systems~\eqref{eq:eqsys} and~\eqref{eq:eqsysMNW} have an intersection but are not equivalent.
\end{remark}

In the paper~\cite{NW} integrable systems of the form
\[
\begin{aligned}
u_t&=v_1,
\\
v_t&=\alpha u_{2k-1}+\beta v_k+g(u,\ldots,u_{2k-2},v,\ldots,v_{k-1})
\end{aligned}
\]
with constant $\alpha$ and $\beta$ were considered. A polynomiality restriction on~$g$ was crucial due to the symbolic approach used in that paper.  A diagonalisation of the linear part (separant) of the system, in symbolic representation, was used to perform a complete classification of homogeneous polynomial systems of that type
with $k=2,3,5$. Some non-polynomial examples of equations~\eqref{eq:eqnev} with $k=3$ were found in~\cite{HSS}.

The symmetry approach \cite{MikShaSok91,ibshab,sokshab} to the classification of integrable scalar evolution equations 
\begin{equation}\label{evsc}
u_t = \phi(u,u_1, \dots , u_n), \qquad u_i = \frac{\partial^i u}{\partial x^i}
\end{equation}
is based on the existence of a formal pseudo-differential recursion operator of the form 
\begin{equation}\label{RR}
R=r_1 D+r_0+r_{-1} D^{-1}+\cdots, \qquad r_i=r_i(u, u_1,\dots)
\end{equation}
satisfying, by definition, the following operator relation 
\begin{equation}\label{eq:mrec}
R_t=[\phi_*,\,R].
\end{equation}
Here 
\[
\phi_*= \sum_{i=0}^n \frac{\partial \phi}{\partial u_i} D^i
\]
is the Fr\'echet derivative of~$\phi$, $D$ is the total $x$-derivative, and the $t$-derivative in~\eqref{eq:mrec} is found by virtue of \eqref{evsc}. A formal recursion operator \eqref{RR} produces \cite{sokshab} an infinite sequence of so-called local {\it canonical conserved densities}
\begin{equation}\label{rhos}
\rho_i={\rm res}\, R^i, \quad i=-1,1,2,\ldots,\qquad\rho_0={\rm res}\,\log(R).
\end{equation}
Requiring locality of the canonical conserved densities leads to efficient necessary integrability conditions upon the right hand side $\phi$ (cf. \cite{CheLeeLiu79}).

Equations~\eqref{evsc} that have an infinite sequence of local symmetries of the form 
\begin{equation}\label{symevsc}
u_{\tau} = \psi(u,u_1, \dots , u_i) 
\end{equation}
possess a formal recursion operator \cite{ibshab,sokshab}. 

If the equation has local conservation laws then, besides formal recursion operators, it admits~\cite{soksvin1} a {\it formal symplectic operator} $S$ such that 
\begin {equation} \label {SSsc}
S_t + S \, \phi_* +\phi_*^ {+} \, S = 0 \,,
\end {equation}
where $^+$ denotes formal operator conjugation. The existence of a formal symplectic operator~$S$ implies that the density $\rho_{2i}$ is trivial for any~$i$.

The series $H=S^{-1}$ satisfies the relation
\begin {equation} \label {HHsc}
H_t - H\phi_*^ {+} -  \phi_*\,H   = 0 \,,
\end {equation}
and is called {\it formal Hamiltonian operator}.

\begin{remark}\label{rem2} Suppose that both formal recursion and symplectic operators $R$ and~$S$ exist. Since $S^{+}$ and $S R^i$ are also formal symplectic operators for any $i\in \Z$, we may assume \cite{sokshab} that $R$ and $S$ are first order series such that
\begin{equation}\label{eq:specr}
R^+=- S\, R\, S^{-1},\qquad {S}^+=-S.
\end{equation}
Since ${\rm res}\,X = -{\rm res}\,X^{+}$ and ${\rm res}\, (Q X Q^{-1}-X) \in {\rm Im}\,D$ for any series $X,Q$, it follows from \eqref{eq:specr} that ${\rm res}\,R^{2i}\in {\rm Im}\,D$ for any $i$.
\end{remark}

\begin{definition}\label{defi1} Equations \eqref{evsc} that have both formal recursion and symplectic operators are called {\it S-integrable}.
\end{definition}
 
In the case of general multi-component evolution systems 
\begin{equation}\label{evsys}
{\bf u}_t = \mathbf{\Phi}({\bf u},{\bf u}_1, \dots , {\bf u}_n),
\end{equation}
where ${\bf u}=(u^1, \dots , u^m)$, and~$\bPhi=(\phi^1,\ldots,\phi^m)$, the Fr\'echet derivative~$\bPhi_*$ is a matrix of differential operators
\begin{equation}\label{eq:frech}
\bPhi_*[i,j]=\sum_{l=0}^n\frac{\partial\phi^i}{\partial u_l^j}D^l=\mathbf{\Sigma}\,D^n+\sum_{l=0}^{n-1}\mathbf{\Phi}^{(l)}D^l,
\end{equation}
where~$\mathbf{\Sigma}$ is the so called \emph{separant} of the evolution system.
Formal recursion~$\R$, symplectic~$\bS$ and Hamiltonian~$\bH$ operators are matrix pseudo-differential series satisfying the equations
\begin{gather}\label{RR2} \R_t=[\mathbf{\Phi}_*,\,\R], \\
\label{SS}  \bS_t + \bS \, \mathbf{\Phi}_* +\mathbf{\Phi}_*^ {+} \, \bS = 0,\\
\label{HH} \bH_t-\bH\mathbf{\Phi}_*^+-\mathbf{\Phi}_*\bH=0
\end{gather}
where~$^+$ denotes matrix transposition followed by the formal operator conjugation of entries. 
\begin{definition}\label{defi2}
A system \eqref{evsys} is called {\it non-degenerate} if its separant matrix $\mathbf{\Sigma}$ is invertible and has no multiple eigenvalues at a generic point.
\end{definition}
The symmetry approach was generalised to the case of non-degenerate systems in~\cite{MSY}. To obtain canonical conserved densities, a diagonalisation procedure was used. This allows to split the matrix relation \eqref{RR2} into~$m$ scalar relations similar to \eqref{eq:mrec} and to apply formula \eqref{rhos}. 
Unfortunately, many of the known integrable evolution systems do not satisfy the non-degeneracy condition. In particular, systems~\eqref{eq:sysgen} are not non-degenerate.

In this paper we generalise the symmetry approach to a class of systems \eqref{eq:eqsys}.
A fundamental point in our work is that for systems~\eqref{eq:eqsys}, albeit being degenerate and non-polynomial, it is possible to perform a full diagonalisation of the Fréchet derivative in the jet space. As a result, a formal recursion operator also reduces to the diagonal form. This allows us to generate explicit integrability conditions in the form of canonical conservation laws. Using these conditions we perform a partial classification of systems of the form
\begin{equation}\label{order4}
\begin{aligned}
u_t&=v,\\
v_t&=u_4+f(u,u_1,\ldots,u_{3},v,v_1)
\end{aligned}
\end{equation}
without the polinomiality restriction for $f$. Most of the obtained equations are not polynomial.

This paper is organised as follows. Section~\ref{sec:2} is devoted to the diagonalisation procedure of systems~\eqref{eq:eqsys}. In Subsection~\ref{subsec:nondeg} we recall results from~\cite{MSY} related to the diagonalisation procedure for non-degenerate systems. 
In Subsection~\ref{subsec:regdiag} we introduce the more general concept of regularly diagonalisable systems. We formulate all statements about such systems without proofs, because they can be proved exactly in the same way as the corresponding statements for the non-degenerate case (cf.~\cite{MSY}).

Section~\ref{sec:3} is devoted to the integrability conditions for systems \eqref{eq:eqsys}. In Section~\ref{sec:4} we present, in the case $k=2$, explicit formulas for several first coefficients of relevant matrix pseudo-differential series appearing in the diagonalisation, as well as the simplest canonical densities.

In Section~\ref{sec:5} we show that any $S$-integrable system \eqref{order4} has the form
\begin{equation}\label{eq:sys41a}
\begin{aligned}
u_t&=v,\\
v_t&=u_4+g_{u_2}u_3^2+2g_{v}v_1u_3+2\left(g_{u_1}u_2+g_uu_1\right)u_3+f_2 v_1^2+f_1 v_1+ f_0,
\end{aligned}
\end{equation}
where~$g=g(u,u_1,u_2,v)$ and $f_i=f_i(u,u_1,u_2,v)$.

In Section~\ref{sec:6} we find all $S$-integrable systems \eqref{eq:sys41a}  with $g=g(u)$ while Section~\ref{sec:7} is devoted to classification of integrable systems with $g=g(u_1)$.
In these sections some finite lists of systems are found. 

Since canonical densities provide only necessary integrability conditions, actual integrability for obtained systems should be justified independently.   One of the possible ways to do that is to construct a compact form for the formal recursion operators, whose first coefficients we initially found in Sections~\ref{sec:5} and \ref{sec:6}. The main tool here is using quasi-local anzatzes for recursion, Hamiltonian and symplectic operators (see \cite{sokkn,malnov,demsok,wang}). Written in this form, a formal recursion operator is a usual recursion operator, and it can be applied to some simple ``seed'' symmetries to generate the hierarchy of infinitesimal symmetries. In Section~\ref{sec:8} we demonstrate this technique applied to the systems found in Section~\ref{sec:5}.

\section{diagonalisation}\label{sec:2}

\subsection{Non-degenerate systems}\label{subsec:nondeg}

For non-degenerate systems~\eqref{evsys} (see Definition~\ref{defi2}), integrability conditions can be obtained through the diagonalisation of the operator $D_t -\mathbf{\Phi}_{*}$, which corresponds to the linearisation of system~\eqref{evsys} around an arbitrary solution~${\bf u}$. Namely, one can show \cite[Proposition 2.1]{MSY} that there exists a matrix pseudo-differential series
\begin{equation}\label{eq:regT}
\T=\T_0+\sum_{l>0}\T_{-l}D^{-l},
\end{equation}
where~$\T_i$ are $m\times m$ matrices depending on $\bu$ and its derivatives, 
such that
\begin{equation}\label{Phhi}
\T^{-1} (D_t -\mathbf{\Phi}_{*}) \T = D_t - 
  \overline{\mathbf{\Phi}}_{*},
\end{equation}
where $ \overline{\mathbf{\Phi}}_{*}=\bT^{-1}\bPhi_*\bT-\bT^{-1}\bT_t$ is a diagonal matrix pseudo-differential series
\begin{equation}\label{Phhi1}
\overline{\mathbf{\Phi}}_{*}={\rm diag}(\Phi_1,\dots,\Phi_m),\qquad \Phi_i=\sum_{l=-n}^{\infty} p_{-l,i} D^{-l},
\end{equation}
such that $p_{n,i}\ne p_{n,j}$ for $i\ne j$ and $ p_{n,i}\ne 0$ for all possible~$i$. In~\eqref{eq:regT}~$\T_0$ is a matrix that diagonalises the separant matrix $\mathbf{\Sigma}$ of system~\eqref{evsys}:
\begin{equation}\label{pp}
\boldsymbol{\Lambda}=\T_0^{-1}\, \mathbf{\Sigma}\, \T_0 =  {\rm diag}(p_{n,1},\dots,p_{n,m}),
\end{equation}
and
\begin{equation}
\label{eq:phib}
 \overline{\mathbf{\Phi}}_* = \boldsymbol{\Lambda}D^n + \sum_{l=-n+1}^{\infty}\overline{\mathbf{\Phi}}^{(-l)}D^{-l}.
\end{equation}
After diagonalisation, formal recursion and symplectic operators satisfy the equations 
\begin{gather}\label{eq:gsymp1}
\overline{\mathbf{R}}_{t}=[\overline{\bPhi}_*,\overline{\mathbf{R}}] , \\
\label{eq:gsymp2}
\overline{\bS}_t + \overline{\bS} \, \overline{\bPhi}_* +\overline{\bPhi}_*^ {+} \, \overline{\bS} = 0,
\end{gather} 
where
\[\overline{\mathbf{R}} =\T^{-1} \mathbf{R} \T, \qquad  \overline{\mathbf{S}} =\T^ {+} \mathbf{S} \T.
\]

The following result can be easily deduced (see~\cite[Theorem 2.1]{MSY}) from~\eqref{eq:gsymp1}:
\begin{proposition} \label{prop1} If a non-degenerate system \eqref{evsys} possesses a formal recursion operator~$\R$,  then the operator $\overline{\mathbf{R}}$ is diagonal:
\begin{equation}\label{Rbar}\overline{\mathbf{R}}={\rm diag}(R_1,\dots,R_m),\qquad R_i=\sum_{l=s_i}^{\infty} r_{-l,i} D^{-l}, \qquad r_{-s_i,i}\neq0,\ s_i\in \mathbb{Z}.
\end{equation}
\end{proposition}
\noindent As both $\overline{\mathbf{\Phi}}_{*}$ and $\overline{\mathbf{R}}$ are diagonal, each of the series $R_i$ satisfies a scalar relation
\begin{equation}\label{drecop}
     (R_i)_t=[ \Phi_i,\,R_i].
\end{equation}
We will call a formal recursion operator~$\mathbf{R}$ \emph{non-degenerate} if in formula \eqref{Rbar} we have $s_i\ne 0$ for all $i$.  Without loss of generality one may assume that in the non-degenerate case the numbers $s_i$ are equal to $-1$. Just as in the scalar case, the functions $\rho_{ij}= {\res}\,R_i^j$ are local conserved densities for~\eqref{evsys}.

The existence a non-degenerate formal recursion operator can be chosen as an integrability criterion for a non-degenerate evolution system~\eqref{evsys}. 
Without imposing the nondegeneracy condition to $\mathbf{R}$, the following problem arises. Suppose that the system consists of two uncoupled systems, one  integrable and the other non-integrable. Then the whole system is not integrable but nevertheless it has a degenerate formal recursion operator.  

Similar problems arise on the level of symmetries or conservation laws: for multi-component systems it is not enough for integrability to have one infinite sequence of symmetries or conservation laws. These sequences must be non-degenerate in some sense. 

A symmetry 
\begin{equation}\label{evsyssym}
{\bf u}_{\tau} = \boldsymbol{\psi}({\bf u},{\bf u}_1, \dots , {\bf u}_k)
\end{equation}
is called \emph{non-singular} if the separant matrix of $\boldsymbol{\psi}$ is non-singular~\cite[p.~10]{MSY}. It can be shown that if the system has an infinite sequence of non-singular symmetries then it possesses a non-degenerate formal recursion operator.

The requirement for nonsingularity of symmetries can be weakened. For instance,  the following statement can be formulated using the diagonalisation procedure. 
It is easy to see that the conjugation by $\T$ diagonalises not only $\mathbf{\Phi}_{*}$ and $\mathbf{R}$, but also the Fr\'echet derivative $\mathbf{\Psi}_{*}$ of  any symmetry
of the system \eqref{evsys}. Namely, we have 
\[
\T^{-1} (D_{\tau} -\mathbf{\Psi}_{*}) \T = D_{\tau} - 
  \overline{\mathbf{\Psi}}_{*},
\]
where $ \overline{\mathbf{\Psi}}_{*}$ is diagonal: 
\[
 \overline{\mathbf{\Psi}}_{*}={\rm diag}(\Psi_1,\dots,\Psi_m),\qquad \Psi_i=\sum_{l=-d_{i}}^{\infty} p_{-l,i} D^{-l},\qquad  p_{d_i,i}\neq0.
\]
The sequence of symmetries 
$$
\mathbf{u}_{\tau_{j}} = \boldsymbol{\psi}_{j},  \qquad j\to \infty
$$
is called {\it non-degenerate} if  $ (d_k)_j \to \infty$ for any $1\le k \le m$.

The following statement can be proved in the same way as Theorem 1.7 from~\cite{sokshab}:

\begin{proposition} \label{prop2} Suppose that a non-degenerate system~\eqref{evsys} has a non-degenerate sequence of symmetries. Then the system has a non-degenerate formal recursion operator.
\end{proposition}

According to Proposition \ref{prop1}, any formal recursion operator becomes diagonal when the Fréchet derivative of the system is already diagonalised. 

The form of symplectic operators of non-degenerate systems~\eqref{evsys} depends on the order~$n$ of the series $\bPhi_{*}$ and the number of equations~$m$ in system \eqref{evsys} (see~\cite{MSY}). Write the series~$\overline{\mathbf{S}}$ in \eqref{eq:gsymp2} as~$\overline{\bS}=\sum_{i=-l}^\infty\overline{\bS}^{(-i)}D^{-i}$. The coefficient of the highest power of $D$, $D^{l+n}$, in~\eqref{eq:gsymp2} implies that
\begin{equation}\label{leadS}(-1)^n\boldsymbol{\Lambda}\overline{\mathbf{S}}^{(l)}+\overline{\mathbf{S}}^{(l)}\boldsymbol{\Lambda}=0,
\end{equation}
where $\boldsymbol{\Lambda}$ is defined by ~\eqref{pp}.  
The equation~\eqref{eq:gsymp2} can be expanded as
\begin{multline*}
\sum_{k\geq0}\left[(-1)^{n}\boldsymbol{\Lambda}\overline{\mathbf{S}}^{(l-k)}+\overline{\mathbf{S}}^{(l-k)}\boldsymbol{\Lambda}\right]D^{l+n-k}+\sum_{i\geq0}\overline{\mathbf{S}}^{(l-i)}_{t}D^{l-i}
\\
+\sum_{k\geq1}\sum_{i=0}^{k-1}\sum_{j=0}^{k-i}\Bigg{[}(-1)^{n-j}\binom{n-j}{k-j-i}D^{k-i-j}\left(\overline{\mathbf{\Phi}}^{(n-j)}\overline{\mathbf{S}}^{(l-i)}\right)
\\
{}+\overline{\mathbf{S}}^{(l-i)}\binom{l-i}{k-j-i}D^{k-j-i}\left(\overline{\mathbf{\Phi}}^{(n-j)}\right)\Bigg]D^{l+n-k}=0.
\end{multline*}
We see that equations from lower coefficients in~\eqref{eq:gsymp2} are of the form
\begin{equation}\label{moreS}(-1)^{n}\boldsymbol{\Lambda}\overline{\mathbf{S}}^{(l-k)}+\overline{\mathbf{S}}^{(l-k)}\boldsymbol{\Lambda}+\boldsymbol{\Delta}_k=0
\end{equation}
where $\boldsymbol{\Delta}_k$ depends on previously found~$\overline{\mathbf{S}}^{(i)}$'s and on~$\overline{\mathbf{\Phi}}^{(i)}$. The terms in~$\boldsymbol{\Delta}_k$ are products~$D^i\overline{\mathbf{\Phi}}^{(j)}\cdot D^r\overline{\mathbf{S}}^{(s)}$ or~$\overline{\mathbf{S}}^{(i)}$, with factors $\overline{\mathbf{S}}^{(i)}$ previously found.
\qquad \begin{enumerate}
    \item \textbf{If~$n$ is odd,} \eqref{leadS} implies that $\overline{\mathbf S}^{(l)}$ is diagonal. Inductively, $\boldsymbol{\Delta}_k$ is also diagonal, which implies that, $\overline{\mathbf{S}}^{(l-k)}$ must be diagonal, and an integrability condition~$\boldsymbol{\Delta}_k=0$.
Thus we can write
\begin{equation}\label{eq:diagS}
    \overline{\mathbf{S}}={\rm diag}(S_1,\dots,S_{m} ),\qquad S_i=\sum_{j=-k_{i}}^{\infty} s_{-j,i} D^{-j}.
\end{equation}
We will call a formal symplectic operator~$\overline{\mathbf{S}}$ in formula \eqref{eq:diagS} \emph{non-degenerate} if  $k_i\ne 0$ for all $i$.   
For systems that have both non-degenerate formal recursion and symplectic operators  we may assume, as in the scalar case, that $k_{i}=1$ and that
\begin{equation}\label{eq:specrv}
\overline{\R}^+=-\overline{\bS}\,\overline{\R}\,\overline{\bS}^{-1},\qquad \overline{\bS}^+=-\overline{\bS}.
\end{equation}
\item 
  \textbf{If $n$ is even,} $n=2k$, it follows from ~\eqref{leadS} that $(\overline{\mathbf{S}}_{l})^i_j \ne 0$ only if $p_{n,i}+p_{n,j}=0$. Rearrange the entries of~$\overline{\mathbf{\Psi}}_{*}$ to obtain a matrix~$\boldsymbol{\Lambda}$ of the form
\begin{equation}\label{newpp}
\operatorname{diag}(p_{n,1},-p_{n,1}, \cdots, p_{n,s}, - p_{n, s}, p_{n, 2s+1}, \cdots, p_{n, 2k}),
\end{equation} 
where $p_{n,i}+p_{n,j}\ne 0$ for $i,j>2s$. Formulas \eqref{moreS} then imply that $\overline{\mathbf{S}}$ has the $2\times 2$ block form
\[\overline{\mathbf{S}}={\rm diag}(\mathbf{S}_1,\dots,\mathbf{S}_s, \mathbf{0}, \dots, \mathbf{0}), \qquad \mathbf{S}_i=\mc{ 0  & \sigma_i\\
\tau_i & 0}
\]
with $\sigma_i, \tau_i$ being scalar pseudodifferential series. If the system has a non-degenerate formal recursion operator, then each non-constant block $\mathbf{S}_i$ can be reduced (see Remark \ref{rem2}) to 
\begin{equation}\label{symcan}
\mathbf{S}_i=\begin{bmatrix} 0  & \sigma_i\\
-\sigma_i^+ & 0
\end{bmatrix}, \qquad \sigma_i=\sum_{j=-1}^{\infty} s_{-j,i} D^{-j}.
\end{equation}
A formal symplectic operator is called {\it non-degenerate} if it has the form 
\[\overline{\mathbf{S}}={\rm diag}(\mathbf{S}_1,\dots,\mathbf{S}_k), 
\]
where the blocks $\mathbf{S}_i$ are of the form~\eqref{symcan}. In this case, if~$p$ is an eigenvalue of~$\boldsymbol{\Lambda}$, then~$-p$ is also an eigenvalue. 
\end{enumerate}

\begin{definition}\label{defi3}  Non-degenerate systems~\eqref{evsys} that have both non-degenerate formal recursion and symplectic operators are called {\it S-integrable}.
\end{definition} 

\begin{remark}\label{rem3}
The existence of an infinite sequence of non-singular\footnote{This means that the leading coefficient of the differential operator $(\delta\rho/\delta\mathbf{u})_{*}$ is non-singular~\cite[p. 10]{MSY}.} conservation laws  
\[D_t\mathbf{\rho}=D\sigma 
\]
implies (cf. \cite{sokshab}) the existence of both non-degenerate formal recursion and symplectic operators.
\end{remark}

\subsection{Regularly diagonalisable systems}\label{subsec:regdiag}

A similar diagonalisation is possible for some degenerate systems \eqref{evsys}. Namely, it can happen that there exists a pseudo-differential matrix series $\T$ such that the series~$\overline{\mathbf{\Phi}}_{*}$ in~\eqref{Phhi} takes the form
\[
\overline{\mathbf{\Phi}}_{*}={\rm diag}(\Phi_1,\dots,\Phi_m),\qquad \Phi_i=\sum_{l=-k}^{\infty} p_{-l,i} D^{-l},
\]
where $1<k<n$, $p_{k,i}\ne p_{k,j}$ for $i\ne j$ and $ p_{k,i}\ne 0$ for all~$i$\footnote{Most statements can be easily generalised to the case when one of the $p_{k,i}$ is equal to zero.}. In this case we say that~\eqref{evsys} admits a~\emph{regular diagonalisation of order} $k$. In \cite[Proposition 2.1]{MSY} it was shown that non-degenerate systems~\eqref{evsys} are regularly diagonalisable of order~$n$. 

Our main observation is that systems \eqref{eq:eqsys} are regularly diagonalisable. 

All the statements from \cite{MSY} about non-degenerate systems reviewed in the previous Section, can be easily generalised to the case of systems admitting a regular diagonalisation. For such systems, Fr\'echet derivatives of symmetries and formal recursion operators become diagonal simultaneously with the Fr\'echet derivative of the system. The definition of the order of a symmetry is the same to that of non-degenerate systems. The canonical form of a formal symplectic operator depends on whether the number~$k$ is even or odd.

\textbf{Recursion and symplectic operators.}
Following the line in \cite[Theorem 2.1]{MSY} it can be proved that, for regularly diagonalisable $S$-integrable systems (see Definition \ref{defi1}) and without loss of generality,  formal recursion and symplectic operators~$\overline{\R}$ and~$\overline{\bS}$ can be assumed to have order~$1$ and be related by the relations~\eqref{eq:specrv}
\[
\overline{\R}^+=-\overline{\bS}\,\overline{\R}\,\overline{\bS}^{-1},\qquad \overline{\bS}^+=-\overline{\bS}.
\]
The matrix series $\overline{\R}$ is always diagonal for regularly diagonalisable systems. The series $\overline{\bS}$ is diagonal if $k$ is odd  and block-diagonal if~$k$ is even, similarly to the non-degenerate case described in Section~\ref{subsec:nondeg}.

\textbf{Integrability conditions.}  After diagonalisation, relation \eqref{RR2} splits into~$m$ scalar equations \eqref{drecop}. These equations define scalar {\it canonical conservation laws}
\begin{equation}\label{eq:rhos}
(\rho_{i,j})_t = (\sigma_{i,j})_x, \qquad i=1,\dots, m, \quad j=-1,0,1,2,\dots ,
\end{equation}
where
\[ 
 \rho_{i,j}={\rm res}\, R_i^j \quad {\rm for}\,\,j\ne 0 \qquad {\rm and} \quad  \rho_{i,0}={\rm res} \, {\rm log} (R_i).
\]
As in the scalar case (see Remark \ref{rem2}), some linear combinations of canonical conserved densities have to be trivial (i.e.~they are total $x$-derivatives). In the next section we describe in detail these integrability conditions in the case of systems~\eqref{eq:eqsys}.

\section{Systems of the Boussinesq type}\label{sec:3}

Consider systems of the form~\eqref{eq:eqsys}.
Their Fréchet derivative and separant are
\begin{equation}\label{eq:sepsys}
{\mathbf{\Phi}}_{*}=\mc{0 & 1\\
U & V},\qquad\mathbf{\Sigma}=\mc{0 & 0\\
f_{u_{2k}} & 0},
\end{equation}
where, denoting~$f_{u_i}=\frac{\partial f}{\partial u_i},$ $f_{v_i}=\frac{\partial f}{\partial v_i}$
\[
U\stackrel{\text{def}} {=}f_{u_{2k}}D^{2k}+\cdots+f_{u}, \qquad V \stackrel{\text{def}} {=} f_{v_{k-1}}D^{k-1}+\cdots+f_{v}.
\] 

\begin{theorem}\label{thm:diagF}
For any system \eqref{eq:eqsys} with $k>1$  
there exists a matrix pseudo-differential series~ $\T$ of the form
\begin{equation}\label{eq:formT}
\mathbf{T}=\mc{1 & T_{2}^{-1}\\
T_{1} & 1},
\end{equation}
being $T_1$ and $T_2$ scalar pseudo-differential series of
order~$k$,  such that 
\[
\T^{-1} (D_t -\mathbf{\Phi}_{*}) \T = D_t - 
  \overline{\mathbf{\Phi}}_{*},  
\]
where
\begin{equation}\label{eq:phibar}
\overline{\bPhi}_{*}
=\begin{bmatrix} f_{u_{2k}}^{1/2} D^{k}+\cdots  & 0\\
0 & -f_{u_{2k}}^{1/2}D^{k} +\cdots
\end{bmatrix}.
\end{equation}
\end{theorem}
\begin{proof}
Substituting
\[
\overline{\bPhi}_*=\mc{\Phi_1 & 0\\
0 & \Phi_2},\qquad
\mathbf{T}=\mc{1 & T_{2}^{-1}\\
T_{1} & 1}
\]
into
\begin{equation}\label{miden}
0=\mathbf{T}\overline{\bPhi}_*-\bPhi_{*}\mathbf{T}+\mathbf{T}_{t},
\end{equation}
we obtain that both $T_1$ and $T_2$ satisfy the Riccati-type operator equation 
\begin{equation}\label{eq:ricc}
T_{t}+T^{2}-VT=U.
\end{equation}
Let us describe the solutions to this equation writing them in the form 
\[T=\tau_k D^k+\tau_{k-1}D^{k-1}+\cdots+\tau_0+\tau_{-1} D^{-1}+\cdots. 
\]
The coefficient of $D^{2k}$ in \eqref{eq:ricc} implies that~$\tau_k^2 = f_{u_{2k}}$. Relations which follow from lower powers of $D$ have the form
\[2\tau_k\tau_j=c(\tau_k,\tau_{k-1},\ldots,\tau_{j+1}),\quad j<k
\]
where the rhs~$c$ is a differential expression depending only on previously found unknowns $\tau_i$,  $i>j$ and the coefficients of~$U$ and~$V$, i.e.~on the rhs of system~\eqref{eq:eqsys}.
Thus, choosing a leading coefficient $\tau_k$, we uniquely find all other coefficients of~$T$. As a result, we get two different solutions of equation \eqref{eq:ricc} with
 leading coefficients $f_{u_{2k}}^{1/2}$ and $-f_{u_{2k}}^{1/2}$. Taking the first solution as~$T_1$ and the second as~$T_2$ leads to formula~\eqref{eq:phibar}.
\end{proof}

\begin{remark}\label{rem:sysvux}
The diagonalisation procedure of Theorem~\ref{thm:diagF} can be applied to many more systems. For example, the systems in Remark~\ref{rem:sdx} can be diagonalised using a matrix~\eqref{eq:formT} with~$T_1$ and~$T_2$ being two different solutions of the operator equation
\[T_t+TD\cdot T-VT=U.
\]
The resulting diagonalised Fréchet derivative takes the form
\[
\overline{\bPhi}_{*}
=\begin{bmatrix} f_{u_{2k-1}}^{1/2} D^{k}+\cdots  & 0\\
0 & -f_{u_{2k-1}}^{1/2}D^{k} +\cdots
\end{bmatrix}
\]
and, thus, this class of systems is regularly diagonalisable.
\end{remark}

\textbf{Recursion, symplectic operators, and integrability conditions.}
As it was mentioned in Paragraphs~2.1 and~2.2, in the case $k=2 l$ the matrix $\overline{\bS}$ consists of $l$ blocks of  size $2\times 2$.
In particular, for systems~\eqref{eq:eqsys} with $k=2$ we have 
\begin{equation}\label{keven}
\overline{\bPhi}_*=\mc{\Phi_1 & 0\\
0 & \Phi_2},  \qquad 
\overline{\R}=\mc{R_1&0\\0& R_2},\qquad \overline{\bS}=\mc{0&S_1\\S_2&0}.
\end{equation}
Considering~\eqref{eq:specrv} it follows that
\begin{equation}\label{eq:eqRS}
\overline{\R}=\mc{R&0\\0&-S^{-1}R^+S},\qquad \overline{\bS}=\mc{0&S\\-S^+&0}.
\end{equation}
where relations~\eqref{eq:gsymp1}, \eqref{eq:gsymp2} imply that 
\begin{equation}\label{eq:sceqRS}R_t=[\Phi_1,R],\qquad S_{t}+\Phi_1^{+}S+S\Phi_2=0.
\end{equation}
Using~\eqref{eq:eqRS} it is easy to prove that
\begin{equation}\label{eq:rhosj}
\rho_{1,j}+(-1)^j \rho_{2,j}\in{\rm Im}\, D, \qquad j=-1,0,1,\dots  .
\end{equation}
Let us denote
\begin{gather*}
\rho_{-1}=\tfrac12(\rho_{1,-1}+\rho_{2,-1}),
\\
\begin{gathered}
\quad\rho_{4i}=\rho_{1,2i}+\rho_{2,2i},\quad\rho_{4i+1}=\rho_{1,2i}-\rho_{2,2i},\\
\rho_{4i+2}=\rho_{1,2i+1}-\rho_{2,2i+1},\quad \rho_{4i+3}=\rho_{1,2i+1}+\rho_{2,2i+1},
\end{gathered}
\qquad i=0,1,2,3,\ldots .
\end{gather*}
According to \eqref{eq:rhos}, all functions~$\rho_{i}$ have to be local conserved densities: ~$D_t\rho_i=D\sigma_i$ for~$i=-1,0,1,\ldots$.
From~\eqref{eq:rhosj} it follows that
the even-numbered densities~$\rho_{2i}$ are divergencies (total $x$-derivatives): ~$\rho_{2i}=D\pi_{2i}$, $i=0,1,2,\ldots$. In this case ~$\sigma_{2i}=D_t\pi_{2i}$. 
The fluxes $\sigma_{2i-1}$ and~potentials~$\pi_{2i}$ are functions appearing in the coefficients of ~$R$ and~$S$.

The relations
\begin{equation}\label{eq:cint}
    \begin{aligned}
        D_t\rho_i&=D\sigma_i,&&\text{if $i$ is $-1$ or odd}\\
        \rho_i&=D\pi_i,&&\text{if $i$ is even}
    \end{aligned}
\end{equation}
form an infinite set of integrability conditions for $S$-integrable systems ~\eqref{eq:eqsys},~$k=2$.

\section{Boussinesq-type systems of fourth order}\label{sec:4}

Let us detail the results and constructions of the previous section in the case of systems~\eqref{eq:eqsys}, $k=2$:
\begin{equation}\label{eq:eqsys4}
u_t=v,\qquad
v_t=f(u,u_1,u_2,u_3,u_4,v,v_1).
\end{equation} 
For such systems $U$ and $V$ are differential operators of orders 4 and 1:
\[
U= \uu_4D^{4}+\uu_3D^{3}+\uu_2D^{2}+\uu_1D+\uu_0,\qquad
V= \vv_{1}D+\vv_{0}
\]
where for conciseness we have denoted $\uu_i=\partial f/\partial u_i$, $\uu_0=\partial f/\partial u$, $\mathfrak{v}_1=\partial f/\partial v_1$ and $\mathfrak{v}_0=\partial f/\partial v$.

Solving~\eqref{eq:ricc} yields that the first terms in~$T_1$ are
\begin{multline*}
T_{1}=\sqrt{\uu_4}\, D^2 + \left[\frac{\uu_3}{2 \sqrt{\uu_4}}+\frac{\vv_1}{2}-\frac{D_{x}\uu_4}{2 \sqrt{\uu_4}}\right]\, D
\\
+ \bigg[-\frac{\uu_3^2}{8 \uu_4^{3/2}}+\frac{3 D \uu_4\, \uu_3}{8 \uu_4^{3/2}}+\frac{\vv_1^2}{8 \sqrt{\uu_4}}-\frac{\left(D \uu_4\right)^2}{8 \uu_4^{3/2}}+\frac{\vv_0}{2}
-\frac{D \uu_3}{2 \sqrt{\uu_4}}
\\
\quad+\frac{\vv_1 D \uu_4}{8 \uu_4}-\frac{1}{2} D \vv_1-\frac{D_{t}\uu_4}{4 \uu_4}+\frac{D^2\uu_4}{4 \sqrt{\uu_4}}+\frac{\uu_2}{2 \sqrt{\uu_4}}\bigg]
{}+\cdots
\end{multline*}
and~$T_2$ can be obtained from~$T_1$ via the replacement $\sqrt{\uu_4}\to -\sqrt{\uu_4}$.  Relation \eqref{miden} leads to
\begin{multline*}
\Phi_1=\sqrt{\uu_4}\, D^2 + \bigg[\frac{\uu_3}{2 \sqrt{\uu_4}}+\frac{\vv_1}{2}-\frac{D\uu_4}{2 \sqrt{\uu_4}}\bigg]\, D
\\
{} + \bigg[-\frac{\uu_3^2}{8 \uu_4^{3/2}}+\frac{3 D\uu_4 \uu_3}{8 \uu_4^{3/2}}+\frac{\vv_1^2}{8 \sqrt{\uu_4}}-\frac{\left(D\uu_4\right)^2}{8 \uu_4^{3/2}}+\frac{\vv_0}{2}-\frac{D\uu_3}{2 \sqrt{\uu_4}}\\
\quad{}+\frac{\vv_1 D\uu_4}{8 \uu_4}-\frac{1}{2} D\vv_1-\frac{D_{t}\uu_4}{4 \uu_4}+\frac{D_{xx}\uu_4}{4 \sqrt{\uu_4}}+\frac{\uu_2}{2 \sqrt{\uu_4}}\bigg]+\cdots
\end{multline*}
The series $\Phi_2$ is related to $\Phi_1$ by the replacement $\sqrt{\uu_4}\to -\sqrt{\uu_4}$.

The first terms of series~$R$ and $S$ defined by relations \eqref{eq:sceqRS} have the form
\begin{multline*}
R=
\uu_4^{1/4}\, D + \left[-\frac{1}{2}\sigma_{-1}+\frac{\vv_1}{4 \uu_4^{1/4}}-\frac{3 D\uu_4}{8 \uu_4^{3/4}}+\frac{\uu_3}{4 \uu_4^{3/4}}\right]
\\
{} + \Bigg[\frac{\sigma_{-1}^2}{8 \uu_4^{1/4}}+\frac{1}{4}D\sigma_{-1}+\frac{\sigma_0}{2 \uu_4^{1/4}}+\frac{\vv_1^2}{32 \uu_4^{3/4}}-\frac{35 \left(D\uu_4\right)^2}{128 \uu_4^{7/4}}+\frac{\vv_0}{4 \uu_4^{1/4}}-\frac{\uu_3 \vv_1}{16 \uu_4^{5/4}}-\frac{3 D\uu_3}{8 \uu_4^{3/4}}
\\
{}+\frac{3 \vv_1D\uu_4}{16 \uu_4^{5/4}}+\frac{3 \uu_3D\uu_4}{8 \uu_4^{7/4}}-\frac{3 D\vv_1}{8 \uu_4^{1/4}}-\frac{D_{t}\uu_4}{8 \uu_4^{5/4}}+\frac{5 D^2\uu_4}{16 \uu_4^{3/4}}+\frac{\uu_2}{4 \uu_4^{3/4}}-\frac{3 \uu_3^2}{32 \uu_4^{7/4}}\Bigg]\, D^{-1}+\cdots
\end{multline*}
and 
\begin{multline}\label{eq:opS}
S=\e^{\pi_0} D + \e^{\pi_0} \left[-\frac{\pi_2}{\uu_4^{1/4}}-\frac{\vv_1}{4 \sqrt{\uu_4}}-\frac{3 D\uu_4}{8 \uu_4}+\frac{\sigma_{-1}}{2 \uu_4^{1/4}}+\frac{\uu_3}{4 \uu_4}\right]
\\
{} + \e^{\pi_0} \Bigg[\frac{\pi_4}{2 \sqrt{\uu_4}}+\frac{\pi_2^2}{2 \sqrt{\uu_4}}+\frac{\sigma_1}{4 \sqrt{\uu_4}}+\frac{\sigma_{-1}^2}{8 \sqrt{\uu_4}}+\frac{\vv_1^2}{32 \uu_4}-\frac{35 (D\uu_4)^2}{128 \uu_4^2}-\frac{3 D\uu_3}{8 \uu_4}+\frac{3 \uu_3 D\uu_4}{8 \uu_4^2}
\\
{}-\frac{D_{t}\uu_4}{8 \uu_4^{3/2}}+\frac{5 D^2\uu_4}{16 \uu_4}+\frac{\uu_2}{4 \uu_4}-\frac{3 \uu_3^2}{32 \uu_4^2}\Bigg] D^{-1} +\cdots
\end{multline}
where the functions~$\sigma_{2i-1}$ and~$\pi_{2i}$, $i=0,1,2,\ldots$ are those in~\eqref{eq:cint}. 

The first canonical densities\footnote{The expressions for the densities~$\rho_i$ are defined up to a total derivative~$\rho_i\to\rho_i+D\alpha_i$; here we choose~$\alpha_i$ to get simpler expressions for the densities, without forgetting to adjust the fluxes correspondingly, $\sigma_i\to\sigma_i+D_t\alpha_i$.} take the form
\begin{gather}
\label{eq:rhom1}
\rho_{-1}=\frac1{\uu^{1/4}},
\end{gather}
\begin{gather}
\rho_0=\frac{1}{2}\frac{\uu_3}{\uu_4}-\frac{3D\uu_4}{4\uu_4},
\end{gather}
\begin{gather}
\rho_1=-\frac{\sigma_{-1}}{4\uu_4}+\frac{1}{2}\frac{\vv_{1}}{\sqrt{\uu_4}},
\end{gather}
\begin{gather}
\rho_2=\frac{\sigma_{0}}{2\uu_4^{1/4}}+\frac{1}{2}D\sigma_{-1}+\frac{3\vv_{1}D\uu_4}{8\uu_4^{5/4}}-\frac{3D\vv_{1}}{4\uu_4^{1/4}}+\frac{\vv_{0}}{2\uu_4^{1/4}}-\frac{\uu_3\vv_{1}}{8\uu_4^{5/4}},
\end{gather}
\begin{gather}
\label{eq:rho3}
\begin{aligned}\rho_3= & \frac{\sigma_1}{2\uu_4^{1/4}}+\frac{\sigma_{-1}^{2}}{4\uu_4^{1/4}}-\frac{D_{t}\uu_4}{4\uu_4^{5/4}}-\frac{35\left(D\uu_4\right)^{2}}{64\uu_4^{7/4}}-\frac{3D\uu_3}{4\uu_4^{3/4}}
\\
 &{}+\frac{3\uu_3D\uu_4}{4\uu_4^{7/4}} +\frac{5D^{2}\uu_4}{8\uu_4^{3/4}}-\frac{3\uu_3^{2}}{16\uu_4^{7/4}}+\frac{\vv_{1}^{2}}{16\uu_4^{3/4}}+\frac{\uu_2}{2\uu_4^{3/4}}.
\end{aligned}
\end{gather}
Subsequent densities can be calculated by computer. Their length grows fast, and we do not present them here.

\section{On the classification of integrable fourth-order Boussinesq-type equations with separant~1}\label{sec:5}

As an application of the previous developments, we consider systems~\eqref{eq:eqsys4} of the form
\begin{equation}\label{eq:sys41}
u_t=v,\qquad
v_t=u_4+f(u,u_1,u_2,u_3,v,v_1)
\end{equation}
and show how to use the integrability conditions for finding the integrable cases.

For systems~\eqref{eq:sys41},   $\uu_4$ is equal to~1, considerably simplifying the canonical densities, being the first ones:
\begin{gather*}
\rho_{-1}=1,\quad\rho_{0}=\frac{\uu_3}{2},\quad\rho_1=\frac{\vv_1}{2},
\quad
\rho_2=\frac{\sigma_0}{2}-\frac{3}{4} D\vv_1-\frac{\uu_3 \vv_1}{8}+\frac{\vv_0}{2},\\
\rho_3=\frac{\sigma_1}{2}-\frac{3}{4} D\uu_3-\frac{3 \uu_3^2}{16}+\frac{\vv_1^2}{16}+\frac{\uu_2}{2},
\\
\begin{split}\rho_4&=\sigma_2-\frac{1}{2}D_{t}\uu_3+\frac{3}{4} \uu_3 D\uu_3-D\uu_2+\frac{1}{2} D^2\uu_3-\frac{1}{4} \vv_1D\vv_1
\\
&\qquad{}-\frac{1}{8} \uu_3 
\vv_1^2+\frac{\uu_3^3}{8}+\frac{\vv_0 \vv_1}{2}-\frac{\uu_2 \uu_3}{2}+\uu_1.
\end{split}
\end{gather*}
We will refer to the integrability condition corresponding to a density~$\rho_i$ as~$C_i$.  

\begin{lemma}\label{ansatz} If system \eqref{eq:sys41} satisfies the first two integrability conditions 
\begin{gather}
C_0:\, \uu_3=2D g(u,u_1,u_2,v)\label{eq:dc0}\\
C_1:\, D_t \frac{\vv_1}{2}=D\sigma_1\label{eq:dc1}
\end{gather}
then it has the form
\begin{equation}\label{eq:sys41eq:sys41b}
\begin{aligned}
u_t&=v,\\
v_t&=u_4+g_{u_2}u_3^2+2g_{v}v_1u_3+2\left(g_{u_1}u_2+g_uu_1\right)u_3+f_2 v_1^2+f_1 v_1+ f_0.
\end{aligned}
\end{equation}
where~$f_i=f_i(u,u_1,u_2,v)$.
\end{lemma}
\begin{proof}
Condition \eqref{eq:dc0} amounts to
\[
f_{u_3}(u,u_1,u_2,u_3,v,v_1)=2Dg(u,u_1,u_2,v)=2g_{u_2}u_3+2g_{v}v_1+2g_{u_1}u_2+2g_uu_1
\]
and, therefore
\[
f=g_{u_2}u_3^2+2g_{v}v_1u_3+2\left(g_{u_1}u_2+g_uu_1\right)u_3+\tilde{f}(u,u_1,u_2,v,v_1).
\]
Now condition~\eqref{eq:dc1} becomes equivalent to
\begin{equation}\label{cond1}
D_t(2g_{v}u_3+\tilde{f}_{v_1})=D\sigma_1(u,u_1,u_2,u_3,u_4,v,v_1,v_2),
\end{equation}
where $\sigma_1$ is an unknown function. By subtracting an appropriate total $x$-derivative from both sides, we can lower the differential order of the left-hand side %\deleted[id=RH]{(and hence the order of the function $\sigma$)} 
until the dependence of the %\replaced[id=RH]{right-hand}{left-hand}
right-hand side on higher derivatives becomes nonlinear. In particular, 
\[
D_t(2g_{v}u_3+\tilde{f}_{v_1})-D(u_4\tilde{f}_{v_1,v_1}+2g_vv_2) 
=-\tilde{f}_{v_1v_1v_1} u_4v_2 +Au_4+Bv_2+C
\]
where $A$, $B$, $C$ are functions of~$(u,u_1,u_2,u_3,v,v_1)$.
Since the rhs has to be a total $x$-derivative, it should be ~$\tilde{f}_{v_1v_1v_1}=0$. 
This proves the lemma and illustrates the general procedure to solve the system of integrability conditions.
\end{proof}
\begin{remark}\label{rem:5} The standard way to eliminate $\sigma_1$ from \eqref{cond1} is applying Euler operators~(cf.~e.g.~\cite{Olv93}) to both sides of the equation. We would obtain
\begin{multline*}
    0=\frac{\delta}{\delta u} \Big(D_t(2g_{v}u_3+\tilde{f}_{v_1})\Big)= -\frac{1}{2} \tilde{f}_{v_1v_1v_1}v_{6}-\frac{5}{2}  \tilde{f}_{v_1v_1v_1v_1}v_{2} v_{5}
    \\
    +\left(g_{u_2} \tilde{f}_{v_1v_1v_1}-\frac52 \tilde{f}_{u_2v_1v_1v_1}\right)u_3 v_{5} 
    +\left(g_{u_2} \tilde{f}_{v_1v_1v_1}-\frac32 \tilde{f}_{u_2v_1v_1v_1}\right) v_{2} u_{6}
    \\
    -  \left(6 g_{vv} g_{u_2}-3 \tilde{f}_{u_2v_1v_1} g_{u_2}-3 g_{u_2vv}+\frac{3}{2}\tilde{f}_{u_2u_2v_1v_1}\right)u_3 u_{6}+\cdots
    \\
    \Rightarrow \tilde{f}_{v_1v_1v_1}=0.
\end{multline*}
We prefer the procedure based on the refinement of $\sigma_1$ for two reasons. First, to use the condition~ $C_3$ we need to know at least the dependence of~$\sigma_1$ on higher derivatives. Elimination of~$\sigma_1$ won't help here. Second, the direct computation of the variational derivative of long expressions is very demanding on computational resources.
\end{remark}

Continuing the classification of integrable systems \eqref{eq:sys41b} requires considering a function $g$ of different types.  In this paper we investigate the cases~$g=g(u)$ and $g=g(u_1)$ (it can be proved that the case~$g=g(u,u_1)$ can be reduced to these two subcases). It turns out that, even in these simple cases, we arrive at non-trivial lists of integrable equations. 

\begin{remark}
The class of systems \eqref{eq:sys41b} is invariant with respect to point transformations
\begin{equation}\label{eq:trsys}
u\to\varphi(p),\qquad v\to\varphi'(p)q.
\end{equation}
A classification of integrable systems \eqref{eq:sys41b} should be performed up to transformations \eqref{eq:trsys}.
It is easy to verify that in the case $g=g(u,u_1)$ the function $g$ transforms as:
\[g(u,u_1)\to2\log\varphi'(p)+g(\varphi(p),\varphi'(p)p_1).
\]
\end{remark}

%======================================================
\section{Classification of the case~$g=g(u)$}\label{sec:6}
%\section{Classification of the case $g=g(u)$}
%======================================================

Using a transformation \eqref{eq:trsys}, we reduce $g$ to zero. Then the system~\eqref{eq:sys41b} becomes
\begin{equation}\label{eq:sysg0}
\begin{aligned}
u_t&=v,\\
v_t&=u_4+f_2 v_1^2+f_1 v_1+ f_0,
\end{aligned}
\end{equation}
with $f_i=f_i(u,u_1,u_2,v)$, $i=0,1,2$.
\begin{theorem}\label{thm:gu}
If a nonlinear equation of the form \eqref{eq:sysg0} satisfies the integrability conditions $C_0$--$C_8$ then it can be reduced, using transformations $t\to \alpha t$, $x\to\beta x+\gamma t$, $u\to \delta u +\kappa x + \lambda t+\xi x^2+\chi t^2$ to one of the following:
\begin{flalign}
\hspace{15pt}&\left\{\begin{aligned}\label{eq:gu1}
u_t&=v\\
v_t&=u_4+ u_2^2
\end{aligned}\right.&
\\[2mm]
&\left\{\begin{aligned}\label{eq:gu2}
u_t&=v\\
v_t&=u_4+ 2u_1u_2
\end{aligned}\right.
\\[2mm]
&\left\{\begin{aligned}\label{eq:gu3}
u_t&=v\\
v_t&=u_4+ 2u_2 v+2u_1^2 u_2
\end{aligned}\right.
\\[2mm]
&\left\{\begin{aligned}\label{eq:gu4}
u_t&=v\\
v_t&=u_4+ 4u_1 v_1+2 u_2 v-6 u_1^2 u_2.
\end{aligned}\right.
\end{flalign}
\end{theorem}
In order to prove Theorem~\ref{thm:gu}, we first perform a preliminary classification producing an ansatz for integrable systems \eqref{eq:sysg0} where the only undetermined expressions are constants. In a second step we continue applying integrability conditions until obtaining a final list of equations \eqref{eq:gu1}-\eqref{eq:gu4} that pass a large (but finite) number of conditions. The third step is to directly prove the integrability of the listed equations by finding a recursion operator for each of them (see~Section~\ref{sec:8}).
\begin{lemma}\label{lem:class}
All integrable systems of the form~\eqref{eq:sysg0} can be written as
\begin{equation}\label{eq:constsys}
\begin{aligned}
u_t&=v\\
v_t&=u_4+(\alpha_0+\alpha_1 u_1) v_1+ 2\left(\beta-\frac14 \alpha_1\right)u_2v +(\delta_0+\delta_1 u) u_2^2
\\
&
\qquad{}+\frac18\left(16\delta_1-\alpha_1^2-12\alpha_1 \beta+16\beta^2\right)u_1^2u_2+(\kappa_0+\kappa_1u_1)u_2 + \tau_0
\end{aligned}
\end{equation}
where~$\alpha_0$, $\alpha_1$, $\beta$, $\delta_0$, $\delta_1$, $\kappa_0$, $\kappa_1$, $\tau_0$ are constants.
\end{lemma}

\begin{proof}[Scheme of the proof of Lemma 2]
It follows from~$C_1$, $C_2$ and~$C_3$ that
\begin{gather}
f_2(u,u_1, u_2,v)=0,\label{eq:f2}
\\
f_1(u,u_1,u_2,v)=a(u,u_1)+f_{10}(u,u_1)v+f_{11}(u,u_1)u_2.\label{eq:f1}
\end{gather}
Now condition~$C_2$ allows to write~$f_0$ in terms of other functions as
\begin{equation}\label{eq:f0a}
f_0(u,u_1,u_2,v)=
\tilde{f}_0\left(u,u_1,u_2\right)+\tfrac{1}{2} \left[(f_{10})_{u_1} u_2 
+(f_{10})_{u} u_1\right]v^2
+2( h_{u_1}u_2 + h_{u}u_1 )v
\end{equation}
where~$h$ is a function of~$u$ and~$u_1$.
Condition~$C_3$ implies that~$f_0(u,u_1,u_2)$ can be written as
\[
\tilde{f}_0(u,u_1,u_2)=f_{00}+u_2 f_{01}+u_2^2 f_{02}-\tfrac{1}{8} u_2^3 f_{10}^2
-\tfrac{1}{24} u_2^3 f_{11}^2
\]
where~$f_{00}$, $f_{01}$ and~$f_{02}$ are functions of~$u$ and~$u_1$. Condition~$C_4$ sets~$f_{11}(u,u_1)=\bar{f}_{11}(u)$ and from~$C_3$ arises the following condition:
\[f_{10}(u,u_1)\bar{f}_{11}(u)=0.
\]
Setting~$\bar{f}_{11}(u)=0$ implies, using~$C_6$, that~$f_{10}(u,u_1)=0$. So we have that anyway~$f_{10}(u,u_1)=0$. Then~$C_1$ is equivalent to 
\[ a(u,u_1)=\alpha_0+a_1(u)u_1+\bar{f}'_{11}(u)u_1^2.
\]
Using in a straightforward way conditions~$C_3$, $C_4$, $C_5$ and~$C_6$ one can find that
\begin{gather*}
\bar{f}_{11}(u)=0,\qquad
    a(u,u_1)=\alpha_0+\alpha_1\,u_1,
\\
    h(u,u_1)=\left(\beta-\tfrac14\alpha_1\right)\,u_1,
\\
    f_{00}(u,u_1)=\bar{f}_{00}(u)
\\
   \begin{split}
   f_{01}(u,u_1)&=-\frac{1}{8} u_{1}^2 \left(-16 \delta_1+\alpha_1^2+12 \alpha_1\, \beta-16 \beta^2\right)+\kappa_0+\kappa_1 u_{1},
   \end{split} 
\\
    f_{02}(u,u_1)=\delta_0+\delta_1 u,
\end{gather*}
where~$\beta$, $\alpha_0$, $\alpha_1$, $\delta_0$, $\delta_1$, $\kappa_0$, $\kappa_1$ are arbitrary constants, 
and the branching condition
\[\beta\cdot \bar{f}_{00}(u)=0.
\]
At this moment the system has the form
\begin{equation*}
\begin{aligned}
u_t&=v\\
v_t&=u_4+(\alpha_0+\alpha_1 u_1) v_1+(\delta_0+\delta_1 u) u_2^2 + 2\left(\beta-\frac14\alpha_1\right)u_2v
\\
&
\qquad{}+\frac18\left(16\delta_1-\alpha_1^2-12\alpha_1 \beta+16\beta^2\right)u_1^2u_2+(\kappa_0+\kappa_1 u_1)u_2 +\bar{f}_{00}(u).
\end{aligned}
\end{equation*}
If we suppose that~$\bar{f}_{00}(u)\neq0$, the conditions up to $C_9$ require that~$\bar{f}_{00}(u)$ is constant, so the lemma is proved.
\end{proof}

\begin{proof}[Scheme of the proof of Theorem~\ref{thm:gu}]
Imposing the integrability conditions up to~$C_9$ to a system of type~\eqref{eq:constsys} yields a system of algebraic equations over the constants~$\beta$, $\alpha_0$, $\alpha_1$, $\delta_0$, $\delta_1$, $\kappa_0$, $\kappa_1$ equivalent to 
\begin{gather*}
    \delta_1=0,\\
    \alpha_1 \delta_0=\kappa_1\delta_0=\beta _1 \delta _0=0,\\
    \alpha _1 \tau _0=\kappa _1 \tau _0=\beta _1 \tau _0=0,\\
    \alpha _1 \left(\alpha _1-2 \beta _1\right)=0,\\
    \alpha_1 \left(3 \alpha_0 \alpha_1+4 \kappa_1\right)=0.
\end{gather*}
Solving the previous system leads to the following three subclasses of systems:
\begin{flalign*}%\label{eq:preboussq}
\hspace{15pt}\left\{
\begin{aligned}
u_t&=v\\
v_t&=u_4+\alpha_0 v_1+\delta_1 u_2^2+\kappa_0 u_2
+\tau_0,
\end{aligned}
\right.&&
\end{flalign*}
\begin{flalign*}%\label{eq:pvu2}
\hspace{15pt}\left\{\begin{aligned}
u_t&=v\\
v_t&=u_4+\alpha_0 v_1+2 \beta v u_2+2 \beta^2 u_1^2 u_2+\kappa_1 u_1 u_2+\kappa_0 u_2,
\end{aligned}\right.&&
\end{flalign*}
\begin{flalign*}%\label{eq:pu1v1}
\hspace{15pt}\left\{\begin{aligned}
u_t&=v\\
v_t&=u_4+\alpha_1 u_1 v_1+\alpha_0 v_1+\frac{1}{2} \alpha_1 v u_2-\frac{3}{8} \alpha_1^2 u_1^2 u_2
-\frac{3}{4} \alpha_0 \alpha_1 u_1 u_2+\kappa_0 u_2.
\end{aligned}
\right.&&
\end{flalign*}
A Galilean transformation $x\to x+\alpha t$ can set~$\alpha_0=0$, and a further change~$u\to u+\xi x^2$ allows to eliminate~$\kappa_0$. The constant~$\tau_0$ can be cancelled by the change~$u\to u+\tau_0 \,t^2/2$, $v\to v+\tau_0\,t$ so the systems become
\begin{flalign}\label{eq:preboussq}
&\hspace{15pt}\left\{\begin{aligned}
u_t&=v\\
v_t&=u_4+\delta_1 u_2^2,
\end{aligned}\right.&
\end{flalign}
\begin{flalign}\label{eq:pvu2}
&\hspace{15pt}\left\{\begin{aligned}
u_t&=v\\
v_t&=u_4+2 \beta v u_2+2 \beta^2 u_1^2 u_2+\kappa_1 u_1 u_2,
\end{aligned}\right.&
\end{flalign}
\begin{flalign}\label{eq:pu1v1}
&\hspace{15pt}\left\{\begin{aligned}
u_t&=v\\
v_t&=u_4+\alpha_1 u_1 v_1+\frac{1}{2} \alpha_1 v u_2-\frac{3}{8} \alpha_1^2 u_1^2 u_2.
\end{aligned}\right.&
\end{flalign}
System~\eqref{eq:preboussq} can be scaled to~\eqref{eq:gu1}. System~\eqref{eq:pvu2} contains two normalized cases: when~$\beta=0$, $u\to \delta u +\kappa x$ leads to~\eqref{eq:gu2}, while when~$\beta\neq0$ a transformation~$x\to\gamma x$, $t\to\gamma^2t$, $u\to \delta u +\kappa x + \lambda t$ leads to~\eqref{eq:gu3}. System~\eqref{eq:pu1v1} can be reduced to~\eqref{eq:gu4} with a transformation $x\to x + \tau t$, $u\to \delta u  + \lambda t$.
\end{proof}

%======================================================
\section{Classification of the case $g=g(u_1)$}\label{sec:7}
%======================================================

When~$g=g(u_1)$ system~\eqref{eq:sys41b} becomes
\begin{equation}\label{eq:sys41b}
\begin{aligned}
u_t&=v,\\
v_t&=u_4+2g'(u_1)u_2u_3+f_2 v_1^2+f_1 v_1+ f_0.
\end{aligned}
\end{equation}
where~$f_i=f_i(u,u_1,u_2,v)$. In this section we consider that~$g'(u_1)\neq0$, to avoid falling into the previous case.

\begin{theorem}\label{thm:gux}
If a nonlinear equation of the form \eqref{eq:sys41b} satisfies the integrability conditions $C_0$--$C_{10}$ then it can be reduced, using transformations $t\to \alpha t$, $x\to\beta x+\gamma t$, $u\to \delta u +\kappa x + \lambda t$   to one of the following
\begin{flalign}
\left\{\begin{aligned}\label{eq:gux1}
u_t&=v\\
v_t&=u_{4}-\frac{4}{u_{1}}u_{2}u_{3}+\frac{1}{u_{1}^{2}}u_{2}v^{2}+\frac{3}{u_{1}^{2}}u_{2}^{3}
\end{aligned}\right.&&
\end{flalign}
\begin{flalign}
\left\{\begin{aligned}\label{eq:gux2}
u_t&=v\\
v_t&=u_{4}-\frac{4}{u_{1}}u_{2}u_{3}+\frac{1}{u_{1}^{2}}u_{2}v^{2}+\frac{3}{u_{1}^{2}}u_{2}^{3}+\left(\epsilon u_{1}^{2}+\frac{1}{u_{1}}\right)u_{2}
\end{aligned}\right.&&
\end{flalign}
\begin{flalign}
&\left\{\begin{aligned}\label{eq:gux3}
u_t&=v\\
v_t&=u_{4}-\frac{4}{u_{1}}u_{2}u_{3}+\frac{1}{u_{1}^{2}}u_{2}v^{2}-\frac{2}{u_{1}^{2}}u_{2}v+\frac{3}{u_{1}^{2}}u_{2}^{3}+\left(\frac{\epsilon}{u_{1}}+\frac{1}{u_{1}^{2}}\right)u_{2}
\end{aligned}\right.&&
\end{flalign}
\begin{flalign}
&\left\{\begin{aligned}\label{eq:gux4}
u_t&=v\\
v_t&=u_{4}-\frac{4}{u_{1}}u_{2}u_{3}+\frac{1}{u_{1}^{2}}u_{2}v^{2}-\frac{2}{u_{1}^{2}}u_{2}v+\frac{3}{u_{1}^{2}}u_{2}^{3}+\left(u_{1}^{2}+\frac{c}{u_{1}}+\frac{1}{u_{1}^{2}}\right)u_{2}
\end{aligned}\right.&&
\end{flalign}
\begin{flalign}
\left\{\begin{aligned}\label{eq:gux6}
u_t&=v\\
v_t&=u_4+\frac{4}{u_{1}} v v_1-\frac{4}{u_{1}}u_{2}u_{3} -\frac{3}{u_{1}^2}u_{2}v^2 +\frac{3}{u_{1}^2}u_{2}^3
\end{aligned}\right.&&
\end{flalign}
\begin{flalign}
\left\{\begin{aligned}\label{eq:gux5}
u_t&=v\\
v_t&=u_{4}+\frac{4}{u_{1}}vv_1+\frac4{u_{1}}v_1-\frac{4}{u_{1}}u_{2}u_{3}-\frac{3}{u_{1}^{2}}u_{2}v^{2}-\frac{2}{u_{1}^{2}}u_{2}v+\frac{3}{u_{1}^{2}}u_{2}^{3}\\
&\quad{}+\left(\epsilon u_{1}^{2}-\frac{1}{3u_{1}^{2}}\right)u_{2}
\end{aligned}\right.&&
\end{flalign}
\begin{flalign}\label{eq:gux7}
\left\{\begin{aligned}
    u_t&=v\\
    v_t&=u_{4}+\frac{4u_1}{u_1^{2}+1}vv_1-\frac{4u_1}{u_1^{2}+1}u_2u_{3}
    	-\frac{3u_1^{2}-1}{\left(u_1^{2}+1\right)^{2}}u_2v^{2}+\frac{3u_1^{2}-1}{\left(u_1^{2}+1\right)^{2}}u_2^{3}
    \\
    &\quad{}+\left(u_1^{2}+1\right)s(u)u_2
    	+\frac{1}{6}s'(u)\left(u_1^{2}+1\right)\left(3u_1^{2}-1\right)
	\\
	&\text{with } s'''(u)+4s(u)s'(u)=0
    \end{aligned}
\right.&&
\end{flalign}
%\begin{flalign}
%&\left\{\begin{aligned}\label{eq:gux8b}
%    u_t&=v\\
%    v_t&=u_4+\frac{4u_1}{u_1^2+1}vv_1-\frac{4 u_1}{u_1^2+1}u_2u_3 -\frac{3 u_1^2-1}{\left(u_1^2+1\right)^2}u_2v^2 
%\\
%&\quad{}
%+\frac{3 u_1^2-1}{\left(u_1^2+1\right)^2} u_2^3+\epsilon \left(u_1^2+1\right) u_2
%\end{aligned}
%\right.&&
%\end{flalign}
\begin{flalign}
&\left\{\begin{aligned}\label{eq:gux8}
    u_t&=v\\
    v_t&=u_4+\frac{4u_1}{u_1^2+1}vv_1+\frac{2 u_1}{u_1^2+1}v_1-\frac{4 u_1}{u_1^2+1}u_2u_3 -\frac{3 u_1^2-1}{\left(u_1^2+1\right)^2}u_2v^2 
\\
&\quad{}-\frac{3 u_1^2-1}{\left(u_1^2+1\right)^2}u_2v 
+\frac{3 u_1^2-1}{\left(u_1^2+1\right)^2} u_2^3-\frac{3 u_1^2-1}{4 \left(u_1^2+1\right)^2} u_2+c \left(u_1^2+1\right) u_2
\end{aligned}
\right.&&
\end{flalign}
%\begin{flalign}
%&\left\{\begin{aligned}\label{eq:gux8}
%u_t&=v\\
%v_t&=u_4+\frac{4u_{1}}{u_{1}^{2}+3}vv_1+\frac43\frac{u_{1}}{u_{1}^{2}+3}v_1-\frac{4u_{1}u_{2}}{u_{1}^{2}+3}u_{3}-\frac{3\left(u_{1}^{2}-1\right)u_{2}}{\left(u_{1}^{2}+3\right)^{2}}v^{2}
%\\
%&\quad{}-\frac{2\left(u_{1}^{2}-1\right)u_{2}}{\left(u_{1}^{2}+3\right)^{2}}v+\frac{3\left(u_{1}^{2}-1\right)}{\left(u_{1}^{2}+3\right)^{2}}u_{2}^{3}+c\left(u_{1}^{2}+3\right)u_{2}-\frac{\left(u_{1}^{2}-1\right)}{3\left(u_{1}^{2}+3\right)^{2}}u_{2}
%\end{aligned}\right.&&
%\end{flalign}
where~$\epsilon$ is equal to~0 or~1 and~$c$ is an arbitrary constant.
\end{theorem}
\begin{proof}[Scheme of the proof of Theorem~\ref{thm:gux}]
Formulas~\eqref{eq:f2} and~\eqref{eq:f1} remain valid in this case. Equation~\eqref{eq:f0a} becomes
\begin{multline*}
f_{0}(u,u_{1},u_{2},v)=
\tilde{f}_{0}(u,u_{1},u_{2})+\left[\left(\tfrac{1}{2}g_{u_1u_1}+\tfrac{1}{4}g_{u_1}f_{10}+\tfrac{1}{2}(f_{10})_{u_{1}}\right)u_2
+\tfrac{1}{2}(f_{10})_{u}u_{1}\right]v^{2}
\\
{}+\left[\tfrac{1}{2}g_{u_1}f_{11}u_{2}^{2}+\tfrac{1}{2}g_{u_1}au_{2}+2h_{u_{1}}u_{2}+2h_{u}u_{1}\right]v
\end{multline*}
and now
\begin{multline*}
    \tilde{f}_{0}(u,u_{1},u_{2})=f_{00}+u_{2}f_{01}+u_{2}^{2}f_{02}-\tfrac{1}{8}f_{10}^{2}u_{2}^{3}-\tfrac{1}{24}f_{11}^{2}u_{2}^{3}
    \\
    {}-\tfrac{1}{4}g_{u_1}f_{10}u_{2}^{3}+\tfrac{1}{2}\left(g_{u_1}\right)^{2}u_{2}^{3}+\tfrac{1}{2}g_{u_1u_1}u_{2}^{3}.
\end{multline*}
When $g_{u_1}(u_{1})\neq0$, as is the case, $C_{4}$ implies that $f_{11}(u,u_{1})=0$. Searching up to~$C_{8}$ we find the following two branch conditions:
\begin{gather*}
\left[f_{10}+2g_{u_1}\right]\left[(f_{10})_{u_{1}}-g_{u_1}f_{10}\right]=0,\\
\left[f_{10}+2g_{u_1}\right]\left[2(f_{10})_{u_{1}}+f_{10}^{2}\right]=0.    
\end{gather*}
\subsection{Case $f_{10}(u,u_1)+2g'(u_1)=0$} 
Condition~$C_1$ implies that~$f_{02}(u,u_1)=0$ and that 
the function~$g(u_1)$ must satisfy the equation  
\[g'''-3g'g''+(g')^2=0
\]
which means that, without loss of generality, we can set
\[g(u_1)=\log{\frac{1}{\gamma_0+(u_1+\gamma_1)^2} }.
\]
Conditions~$C_3$, $C_4$ imply that we can set
\[
h(u,u_1)=\frac14a(u,u_1).
\]
Now two subcases must be considered, $\gamma_0=0$ and $\gamma_0\neq0$.
\subsubsection{Subcase $\gamma_0=0$}
From $C_1$
\[a(u,u_1)=\alpha_0+\frac{b(u)}{u_1+\gamma_1}
\]
and from~$C_5$
\[
f_{01}(u,u_1)=c(u)(u_1+\gamma_1)^2-\frac{\alpha_0}4\frac{b(u)}{u_1+\gamma_1}-\frac3{16}\frac{b(u)^2}{(u_1+\gamma_1)^2}.
\]
Now~$C_1$ imposes that
\begin{multline*}
f_{00}\left(u,u_1\right)=-\frac{1}{6} c'(u) u_{1} \left(\gamma_1+u_{1}\right) \left(5 \gamma_1^2-6 \gamma_1 u_1-3 u_1^2\right)
\\
{}+\frac{1}{4} \alpha_0 b'(u) u_{1}+\frac{3 b(u) b'(u) u_{1}}{16 \left(\gamma_1+u_{1}\right)}.
\end{multline*}
Again a branch appears, with~$\gamma_1=0$ or~$\gamma_1\neq0$.\\

\textbf{When $\gamma_1=0$} all conditions up to $C_{10}$ are solved only if
\begin{equation}\label{eq:f01f100B}
    b(u)c'(u)+2b'(u)c(u)+2b'''(u)=0.
\end{equation}
The system has the form
\begin{multline*}\begin{aligned}
    u_t&=v\\
    v_t&=u_4+ \left(\alpha_0+\frac{b(u)}{u_1}+\frac{4 v}{u_1}\right)v_1-\frac{4 u_2}{u_1}u_3-\frac{3 u_2}{u_1^2}v^2 
    \\
    &\quad{}-\left(\frac{\alpha_0}{u_1}+\frac32\frac{b(u)}{u_1^2}\right)u_2v +\frac{1}{2} b'(u) v
    +\frac{3 u_2^3}{u_1^2}+\left(c(u) u_1^2-\frac{\alpha_0 b(u)}{4 u_1}-\frac{3 b(u)^2}{16 u_1^2}\right)u_2 
    \\
    &\quad{}+\frac{1}{16} \left(8 c(u)_{u} u_1^4+4 \alpha_0 b'(u) u_1+3 b(u) b'(u)\right)
    \end{aligned}
\end{multline*}
A Galilean transformation eliminates~$\alpha_0$. If~$b(u)=0$, a transformation~$\varphi(u)\to u$ with
\[2\frac{\varphi'''(u)}{\varphi'(u)}-3\frac{\varphi''(u)^2}{\varphi'(u)^2}=c(u),
\]
produces equation~\eqref{eq:gux6}. If~$b(u)\neq0$ a transformation~$\varphi(u)\to u$ with~$\varphi'(u)=1/b(u)$, condition~\eqref{eq:f01f100B} and appropriate scalings lead to equation~\eqref{eq:gux5}.  

\textbf{When $\gamma_1\neq0$} we have that $C_1$ implies that~$b(u)$ and~$c(u)$ are constant,
which means that the expression of the system has no function explicitly dependent on~$u$. Then we can perform a transformation~$u\to u-\gamma_1 x$ and set~$\gamma_1=0$, so we fall into the case studied immediately above.

\subsubsection{Subcase $\gamma_0\neq0$}
We normalize $\gamma_0=1$. From $C_{1}$ and $C_{3}$
\[a(u,u_1)=\alpha_0\frac{1-(u_1+\gamma_1)^{2}}{1+(u_1+\gamma_1)^{2}}+\alpha_1\frac{2\left(\gamma_1+u_{1}\right)}{1+(u_1+\gamma_1)^{2}}.
\]
Now from $C_{1}$ we see that
\begin{gather*}
f_{01}(u,u_{1})=\frac{1+(u_1+\gamma_1)^{2}}{2(u_1+\gamma_1)}\left(\sigma_{1}(u,u_1)\right)_{u_1},\\
f_{00}(u,u_{1})=\frac{1+(u_1+\gamma_1)^{2}}{2(u_1+\gamma_1)}\left(\sigma_{1}(u,u_1)\right)_{u}u_1    
\end{gather*}
being $\sigma_{1}(u,u_1)$ an arbitrary function. From $C_{3}$ a branching condition arises
\[\alpha_1\left(\sigma_{1}(u,u_1)\right)_{u}=0.
\]
\indent\textbf{When $\alpha_1=0$}, from conditions~$C_3$ and $C_{5}$
it follows that
\[
    \sigma_{1}(u,u_1)=-\frac{\alpha_0^{2}\left(1+3\left(u_1+\gamma_1\right)^{2}\right)}{4\left(1+\left(u_1+\gamma_1\right)^{2}\right)^{2}}+\frac{1}{2}\sigma(u)\left(\left(u_1+\gamma_1\right)^{2}-\frac13
    \right)+\tau_0
\]
and another branching condition
\[\gamma_1\cdot\sigma'(u)=0.
\]
\indent If $\gamma_1=0$, $C_{5}$ implies that
\begin{equation}
    \sigma'''+2\sigma\sigma'=0
\end{equation}
and all the conditions up to $C_{8}$ are satisfied. The resulting equation is
\begin{multline*}\begin{aligned}
    u_t&=v\\
    v_t&=u_{4}+\frac{4u_1}{u_1^{2}+1}vv_1-\frac{\alpha_0\left(u_1^{2}-1\right)}{u_1^{2}+1}v_1-\frac{4u_1}{u_1^{2}+1}u_2u_{3}
    \\
    &\quad{}-\frac{3u_1^{2}-1}{\left(u_1^{2}+1\right)^{2}}u_2v^{2}+\frac{\alpha_0u_1\left(u_1^{2}-3\right)}{\left(u_1^{2}+1\right)^{2}}u_2v+\frac{\left(3u_1^{2}-1\right)u_2^{3}}{\left(u_1^{2}+1\right)^{2}}
    \\
    &\quad{}+\frac{\alpha_0^{2}(3u_1^{2}-1)}{4\left(u_1^{2}+1\right)^{2}}u_2+\frac12\left(u_1^{2}+1\right)\sigma(u)u_2
    \\
    &\quad{}+\frac{1}{12}\sigma'(u)\left(u_1^{2}+1\right)\left(3u_1^{2}-1\right)
    \end{aligned}
\end{multline*}
that after Galilei and~$\sigma(u)\to2s(u)$ becomes
\begin{multline*}\begin{aligned}
    u_t&=v\\
    v_t&=u_{4}+\frac{4u_1}{u_1^{2}+1}vv_1-\frac{4u_1}{u_1^{2}+1}u_2u_{3}
    	-\frac{3u_1^{2}-1}{\left(u_1^{2}+1\right)^{2}}u_2v^{2}+\frac{\left(3u_1^{2}-1\right)u_2^{3}}{\left(u_1^{2}+1\right)^{2}}
    \\
    &\quad{}+\left(u_1^{2}+1\right)s(u)u_2
    	+\frac{1}{6}s'(u)\left(u_1^{2}+1\right)\left(3u_1^{2}-1\right)
    \end{aligned}
\end{multline*}
with~$\sigma'''+4\sigma\sigma'=0$. This is eq.~\eqref{eq:gux7}.
%
%and some scalings can be written as
%\begin{equation}\label{eq:pregux7}
%\begin{aligned}
%u_t&=v\\
%v_t&=u_4+\frac{4 u_1}{u_1^2+3}v v_1-\frac{4 u_1 u_2 u_3}{u_1^2+3}-\frac{3  \left(u^2_1-1\right) u_2v^2}{\left(u_1^2+3\right)^2}
%\\
%&\quad{}+\frac{3 \left(u_1^2-1\right)u_2^3}{\left(u_1^2+3\right)^2}
%+ s(u) \left(u_1^2+3\right) u_2+\frac{1}{2} s'(u) \left(u_1^2-1\right) \left(u_1^2+3\right)
%\end{aligned}
%\end{equation}
%with~$s'''(u)+4s(u)s'(u)=0$. This is eq.~\eqref{eq:gux7}.

%\addr{There is another version of~\eqref{eq:gux7}
%\begin{multline*}\begin{aligned}
%    u_t&=v\\
%    v_t&=u_{4}+\frac{4u_1}{u_1^{2}-1}vv_1-\frac{4u_1}{u_1^{2}-1}u_2u_{3}
%    -\frac{3u_1^{2}+1}{\left(u_1^{2}-1\right)^{2}}u_2v^{2}+\frac{\left(3u_1^{2}+1\right)u_2^{3}}{\left(u_1^{2}-1\right)^{2}}
%    \\
%    &\quad{}-2\left(u_1^{2}-1\right)s(u)u_2
%    -\frac{1}{3}s'(u)\left(u_1^{2}-1\right)\left(3u_1^{2}+1\right)
%    \end{aligned}
%\end{multline*}
%with $s'''(u)-8s(u)s'(u)=0$. This corresponds to the equation as I wrote it in the Lagrangian classification. The main questions are: $u\to{\rm i}u$ and the scaling of~$\sigma$.
%}

If $\gamma_1\neq0$, we have that~$\sigma(u)=\text{constant}$ and no more restrictions up to level~$C_8$. The resulting equation, with Galilei,  scaling and $u\to u+\alpha x$ transformations, can be reduced to~\eqref{eq:gux7} with constant~$s(u)$.

\indent\textbf{When $\alpha_1\neq0$}, we have that~$\sigma_1(u,u_1)=\sigma(u_1)$ and Galilei, scaling and $u\to u+\alpha x$ transformations lead to 
\begin{multline*}
\begin{aligned}
    u_t&=v\\
    v_t&=u_4+\frac{4u_1}{u_1^2+1}vv_1+\frac{2 \alpha_1 u_1}{u_1^2+1}v_1-\frac{4 u_1 u_2}{u_1^2+1}u_3 -\frac{\left(3 u_1^2-1\right) u_2}{\left(u_1^2+1\right)^2}v^2 
\\
&\quad{}-\frac{\alpha_1\left(3 u_1^2-1\right) u_2}{\left(u_1^2+1\right)^2} v 
+\frac{3 u_1^2-1}{\left(u_1^2+1\right)^2} u_2^3-\frac{\alpha_1^2 \left(3 u_1^2-1\right)}{4 \left(u_1^2+1\right)^2} u_2+\sigma_0 \left(u_1^2+1\right) u_2
\end{aligned}
\end{multline*}
Only scalings are available to simplify the equation now, and we obtain a subcase of~\eqref{eq:gux7} and eq.~\eqref{eq:gux8}.

\subsection{Case $f_{10}(u,u_1)+2g_{u_1}(u_1)\neq0$} We have that~$f_{10}(u,u_1)=0$. Conditions up to $C_6$ imply that~$f_{02}(u,u_1)=0$ and that~$g(u_1)$ must satisfy
\[2g''-(g')^2=0
\]
i.e.~we can take
\[
g(u_1)=-2\log(u_1+\gamma_1).
\]
Conditions up to~$C_7$ imply that
\begin{gather*}
    a(u,u_1)=\alpha_0\text{ constant},
\\
    h(u,u_1)=\frac{d(u)}{u_1+\gamma_1}.
\end{gather*}
The form of the system at this stage is
\[
\begin{aligned}
    u_t&=v\\
    v_t&=u_4+\alpha_0 v_1-\frac{4 u_3 u_2}{u_1+\gamma_1}+\frac{v^2 u_2}{\left(u_1+\gamma_1\right)^2}-\frac{\alpha_0 v u_2}{u_1+\gamma_1}-\frac{2 d(u) v u_2}{\left(u_1+\gamma_1\right)^2}
    \\
    &\quad{}+\frac{3 u_2^3}{\left(u_1+\gamma_1\right)^2}+\frac{2 d'(u) v u_1}{u_1+\gamma_1}+u_2 f_{01}\left(u,u_1\right)+f_{00}\left(u,u_1\right).
\end{aligned}
\]
Condition~$C_3$ implies that
\[
f_{01}\left(u,u_1\right)_{u_1u_1}-\frac{2 f_{01}\left(u,u_1\right)}{\left(u_1+\gamma_1\right)^2}-\frac{4 d(u)^2}{\left(u_1+\gamma_1\right)^4}=0
\]
so we write
\[
f_{01}(u,u_1)=c(u) \left(u_1+\gamma_1\right)^2+\frac{b(u)}{u_1+\gamma_1}+\frac{d(u)^2}{\left(u_1+\gamma_1\right)^2}.
\]
Besides we obtain
\[
f_{00}(u,u_1)=\frac{1}{2} c'(u) \left(u_1-\gamma_1\right) \left(u_1+\gamma_1\right)^3-\frac{1}{2} b'(u) \left(2u_1+\gamma_1\right)-d(u) d'(u)
\]
and the condition
\[
\gamma_1\cdot d'(u)=0.
\]
Supposing that~$\gamma_1\neq0$ implies that~$d(u)$ is constant and, by means of~$C_4$, that also~$b(u)$ and~$c(u)$ are constant. That makes the rhs of the system not explicitly dependent on~$u$, and thus a transformation~$u\to u-\gamma_1x$ can set~$\gamma_1$ to zero, without introducing explicit dependencies on~$x$. We can thus set~$\gamma_1=0$, yielding
\begin{gather*}
f_{01}(u,u_{1})=c(u)u_{1}^{2}+\frac{b(u)}{u_{1}}+\frac{d(u)^{2}}{u_{1}^{2}},\\
f_{00}(u,u_{1})=\frac{1}{2}c'(u)u_{1}^{4}-b'(u)u_{1}-d(u)d'(u).
\end{gather*}
Condition~$C_3$ implies a relation
\[
d(u)b'(u)-d'(u)b(u)=0
\]
opening two subcases, $d(u)=0$ and~$d(u)\neq0$.
\subsubsection{Case $d(u)\neq0$} We have $b(u)=\beta d(u)$ and
\[d(u)c'(u)+2d'(u)c(u)+2d'''(u)=0
\]
so we can write
\[c(u)=\frac{\gamma}{d^2(u)}-\frac{2d(u)d''(u)-(d'(u))^2}{d^2(u)}
\]
We obtain thus a system whose rhs depends on one arbitrary function~$d(u)$ and some arbitrary constants, and that passes all integrability conditions up to~$C_{10}$. A transformation~$\varphi(u)\to u$ with~$\varphi'(u)=1/d(u)$ can set~$d(u)=1$ yielding, after some scalings and Galilei, the system
\begin{equation}\label{eq:pregux123}
\begin{aligned}
u_t&=v\\
v_t&=u_{4}-\frac{4u_{3}u_{2}}{u_{1}}+\frac{v^{2}u_{2}}{u_{1}^{2}}-\frac{2vu_{2}}{u_{1}^{2}}+\frac{3u_{2}^{3}}{u_{1}^{2}}+\left(\kappa_0u_{1}^{2}+\frac{\kappa_1}{u_{1}}
+\frac{1}{u_{1}^{2}}\right)u_{2}
\end{aligned}
\end{equation}
Different normalizations and scalings produce equations~\eqref{eq:gux3} and~\eqref{eq:gux4}.

\subsubsection{Case $d(u)=0$} Condition~$C_5$ yields again the equation 
\[b(u)c'(u)+2b'(u)c(u)+2b'''(u)=0
\]
which can be solved, if~$b(u)\neq0$, in~$c(u)$. The resulting system can be again transformed through~$\varphi(u)\to u$ with~$\varphi'(u)=1/b(u)$ into~\eqref{eq:gux2}. If~$b(u)=0$, a Galilean transformation and~$\varphi(u)\to u$ with $2\frac{\varphi'''(u)}{\varphi'(u)}-3\frac{\varphi''(u)^{2}}{\varphi'(u)^{2}}=c(u)$~allows to get the system~\eqref{eq:gux1}.
\end{proof}
\begin{remark}
In~\cite{CH} a list of integrable Lagrangian systems with Lagrangian density $\mathscr{L}=\frac12L_2(u,u_1,u_2)u_t^2+L_1(u,u_1,u_2)u_t+L_0(u,u_1,u_2)$ was found. Comparing with that list, we find that~Eq.~\eqref{eq:gu3} is equivalent to the Lagrangian system~L3, Eq.~\eqref{eq:gu4} to~L2 and Eq.~\eqref{eq:gux7} to~L8.
\end{remark}

%======================================================
\section{Recursion, symplectic and Hamiltonian operators for the equations in Theorem~\protect\ref{thm:gu}}
\label{sec:8}
%======================================================

The systems found in Theorems~\ref{thm:gu} and~\ref{thm:gux} satisfy only a finite number of necessary conditions for integrability~\eqref{eq:cint}. But integrability requires the satisfaction of an infinite number of integrability conditions. To prove integrability we will determine  explicit recursion and symplectic operators for each of the listed systems. 

Recall that there are the following algebraic relations between operators~\cite{Dorf}. For any recursion~$\R$, symplectic~$\bS$ and~Hamiltonian~$\bH$ operators, the product $\bS\R$ is a symplectic operator,  the product~$\mathbf{R}\mathbf{H}$ is a Hamiltonian operator and the inverse~$\mathbf{S}^{-1}$ is a Hamiltonian operator (and viceversa, the inverse~$\bH^{-1}$ is a symplectic operator).

\textbf{Weakly nonlocal operators.}  
Almost all known recursion, symplectic and Hamiltonian operators for various integrable systems can be written in quasi-local (or weakly nonlocal) form.  
\begin{definition}
    A matrix pseudo-differential operator $\mathbf{A}$ is said to be written in a quasi-local form if
    \begin{equation} \label{sc}
        \mathbf{A}=\mathbf{D}+\sum_{i=1}^k\mathbf{B}_iD^{-1}\cdot \mathbf{C}_i^+
    \end{equation}
    where $\mathbf{D}$ is a differential operator and~$\mathbf{B}_i$, $\mathbf{C}_i$ are  rectangular matrices of differential functions of compatible size.
\end{definition}

The matrices ~$\mathbf{B}_i$, $\mathbf{C}_i$ are related \cite{sokkn,malnov,wang,demsok} to local symmetries and cosymmetries of the corresponding integrable system. Recall \cite{Olv93} that symmetries and cosymmetries are vector-columns $\bf b$ and $\bf c$, depending on jet variables, satisfying respectively the equations
\[
\mathbf{b}_t=\mathbf{\Phi}_*(\mathbf{b}),\qquad \mathbf{c}_t=-\mathbf{\Phi}_*^+(\mathbf{c}).
\]
Here the $t$-derivative is calculated in virtue of the system under consideration. 
A special type of cosymmetry is the variational derivative of a conserved density~$\rho$, that is of the form  
\[\mathbf{c}=\frac{\delta\rho}{\delta{\mathbf{u}}}=\begin{bmatrix}\delta\rho\big/\delta u\\
\delta\rho/\delta v
\end{bmatrix},
\]
where
\[\frac{\delta\rho}{\delta u}=\sum_{i\geq0}D^i\left(\frac{\partial \rho}{\partial u_i}\right),\quad \frac{\delta\rho}{\delta v}=\sum_{i\geq0}D^i\left(\frac{\partial \rho}{\partial v_i}\right).
\]

In the case of recursion operators, in formula~\eqref{sc} the columns of matrices 
$\mathbf{B}_i$ and~$\mathbf{C}_i$ are correspondingly symmetries and cosymmetries. For symplectic operators both $\mathbf{B}_i$ and~$\mathbf{C}_i$ consist of cosymmetries,  and for Hamiltonian operators both~$\mathbf{B}_i$ and~$\mathbf{C}_i$ are formed by symmetries.

Quasi-local anzats are practical because finding symmetries and cosymmetries of a system is relatively simple.  If the rhs of a system as well as its symmetries and cosymmetries are homogeneous polynomials, computations become almost trivial. 

In the simplest version of the ansatz \eqref{sc} the matrices $\mathbf{B}_i$ and $\mathbf{C}_i$ are column-vectors. We will call such operators {\it vector quasi-local operator}. In particular, 
vector quasi-local   recursion operator~$\mathbf{R}$ has the form
\begin{equation}\label{eq:quasilr}
    \mathbf{R}=\mathbf{D}+\sum_{i=1}^k\mathbf{s}_iD^{-1}\cdot \mathbf{c}_i^+,
\end{equation}
where $\mathbf{s}_i$ are  symmetries, and $\mathbf{c}_i$ are  cosymmetries of the integrable system. 

We will find vector quasi-local recursion and symplectic operator for Boussinesq type systems from Theorem \ref{thm:gu}. For each  of systems \eqref{eq:gu1}, \eqref{eq:gu2} and \eqref{eq:gu4} we have found two Hamiltonian operators and present a bi-Hamiltonian for of the system. For system \eqref{eq:gu3} we also found two Hamiltonian operators. However, we were able to write in the bi-Hamiltonian form a higher symmetry, but not the system itself.

\subsection*{System~\eqref{eq:gu2}} 
This system admits the symmetries $\mathbf{s}_1=\mc{1\\0}$, $\mathbf{s}_2=\mc{u_1\\v_1}$  and cosymmetries $\mathbf{c}_1=\mc{0 \\1}$, $\mathbf{c}_2=\mr{-v_1\\u_1}$, variational derivatives of conserved densities~$v$ and~$u_1v$ respectively. 
Using these symmetries and cosymmetries we can write recursion, symplectic and Hamiltonian operators $\R , \bS_1, \bS_2, \bH_1 $ and $\bH_2$ in quasi-local form as follows:
\[
\R=\mc{\frac{3}{4}v&D\\[1em]
 R_{21}& \frac34v}
 +\frac{1}{4}\begin{bmatrix}u_1\\v_1
\end{bmatrix}D^{-1}\cdot \begin{bmatrix}0 & 1
\end{bmatrix}
+\frac{1}{4}\mc{1\\0}D^{-1}\cdot\mc{-v_1 & u_1},
\]
where $R_{21}=D^{5}+\frac{5}{2}u_1D^{3}+\frac{15}{4}u_2D^{2}+\left(\frac{9}{4}u_3+u_1^{2}\right)+\frac{1}{2} u_4+u_1u_2$,
\[
\bS_1=\mr{0 & -1\\
1 & 0},\qquad\mathbf{H}_{1}=\mathbf{S}_{1}^{-1}=\mr{0 & 1\\
-1 & 0
}
\]
and
\begin{gather*}
    \bS_2=\bS_1\R
=
\mc{-R_{21}
& -\frac{3}{4}v\\[1.5em]
\frac{3}{4}v& D
}+\frac{1}{4}\mr{-v_1\\
u_1}D^{-1}\cdot \begin{bmatrix}0 & 1
\end{bmatrix}
+\frac{1}{4}\mc{0\\1}D^{-1}\cdot\mc{-v_1 & u_1},
\\
\mathbf{H}_{2}=-\R\bH_{1}=\begin{bmatrix}D & -\frac{3}{4}v\\[1em]
\frac{3}{4}v & -R_{21}
\end{bmatrix}
+\frac14\mc{u_1\\v_1}D^{-1}\cdot\mc{1&0}
+\frac14\mc{1\\0}D^{-1}\cdot\mc{u_1&v_1}.
\end{gather*}
The Hamiltonian operators $\bH_1$ and $\bH_2$ are compatible \cite{Olv93}. A bi-Hamiltonian pencil can be obtained from~$\mathbf{H}_{2}$ by shifting $v\to v+\text{const}$
and two corresponding Hamiltonian forms of the system are given by
\[
\begin{bmatrix}u\\
v
\end{bmatrix}_{t}=\bH_1\frac{\delta}{\delta\mathbf{u}}\left(\frac{1}{2}v^{2}-\frac{1}{2}u_{2}^{2}+\frac{1}{3}u_{1}^{3}\right)=\bH_2\frac{\delta}{\delta\mathbf{u}}\left(-2v\right).
\]

\subsection*{System~\eqref{eq:gu4}} This system admits the symmetry~$(1,0)$ and the conserved density~$\rho=v-u_1^2$, with variational derivative~$(2u_2,1)$, allowing us to write the quasi-local expression for a recursion operator
\[
\R=\mc{-u_1&0\\D^{3}+\left(2v-6u_{1}^{2}\right)D+v_{1} & 3u_1}
+\mc{1\\0}D^{-1}\cdot \mc{2u_2& 1}.
\]

Two symplectic operators are
\begin{equation}\label{eq:so51a}
\mathbf{S}_{1}=\begin{bmatrix}4u_1D+2u_2 & -1\\
1 & 0
\end{bmatrix}
\end{equation}
and
\begin{equation}\label{eq:so51b}
\mathbf{S}_{2}=\bS_1\R=\begin{bmatrix}-D^{3}+2(u_1^{2}-v)D+\left(2u_1u_2-v_1\right)& u_1\\
-u_1 & 0
\end{bmatrix}
+\begin{bmatrix}2u_2\\
1
\end{bmatrix}D^{-1}\cdot \begin{bmatrix} 2u_2& 1 \end{bmatrix}.
\end{equation}

Two Hamiltonian operators have the form
\[
    \mathbf{H}_{1}=\mathbf{S}_{1}^{-1}=\mc{\phantom{-}0 & 1\\-1 & 4u_{1}D+2u_{2}}
\]
and  
\[
\bH_2=-\R\mathbf{H}_1=\mc{0 & -3u_1\\3u_1 & -D^{3}-\left(2v+6u_{1}^{2}\right)D-v_{1}-6u_1u_2}
+\mc{1\\0}D^{-1}\cdot \mc{1& 0}
\]
These Hamiltonian operators are related by the argument shift $u\to u+\text{const}\, x$. We can write the system in bi-Hamiltonian form as
\[\begin{bmatrix}u\\
v
\end{bmatrix}_{t}=\bH_1\frac{\delta}{\delta\mathbf{u}}\left(\frac{1}{2}v^{2}-\frac{1}{2}u_{2}^{2}-\frac{1}{2}u_{1}^{4}\right)=\bH_2\frac{\delta}{\delta\mathbf{u}}\left(-u_{1}v+u_{\text{1}}^{3}\right).
\]
\subsection*{System~\eqref{eq:gu3}} The simplest symmetries and cosymmetries of this system are
\[\mathbf{s}_1=\mc{1\\0},\ \mathbf{s}_2=\mc{v\\v_t},\ \mathbf{s}_3=\mc{u_1\\v_1},\quad\mathbf{c}_1=\mc{-2u_2 \\1},\ \mathbf{c}_2=\mc{-4u_1u_2-2v_1\\2u_1},
\]
where the latter two are the variational derivatives of conserved densities~~$v+u_1^2$, and~$\frac23u_1^3+2u_1v$. 
Using those we find the quasi-local recursion operator
\begin{multline*}
\R=\mc{\frac{1}{2}u_{2}D+\frac{1}{2}u_{1}^{3}+\frac{3}{2}u_{1}v&D
\\[1em]
R_{21}
&\frac52u_2D+\frac12u_1^3+\frac32u_1v+2u_3
}
\\
{}+\frac{1}{2}
\begin{bmatrix}v\\u_4+ 2u_2 v+2u_1^2 u_2\end{bmatrix}D^{-1}\cdot \mc{-2u_2&1}
\\
{}+\frac{1}{2}\mc{u_1\\v_1}D^{-1}\cdot\mc{-2u_1u_2-v_1&2u_1}
.
\end{multline*}
with
\begin{multline*}
    R_{21}=D^5+\frac{5}{2}(u_{1}^{2}+v)D^{3}+(\frac{15}{2}u_{1}u_{2}+2v_{1})D^{2}
\\\textstyle
{}+(u_{4}+2u_{1}^{2}v+v^{2}+\frac{9}{2}u_{1}u_{3}+4u_{2}^{2}+\frac{1}{2}v_{2})D
\\\textstyle
{}+\frac{1}{2}u_{1}^{2}v_{1}+2u_{1}u_{2}v+\frac{1}{2}vv_{1}+2u_{1}^{3}u_{2}+u_{1}u_{4}
\end{multline*}
and the symplectic operator
\begin{multline*}
\bS=\mc{
S_{11}
&-D^2-\frac32u_1^2-\frac12v
\\
D^{2}+\frac{3}{2}u_{1}^{2}+\frac{1}{2}v&0}
\\
{}+\frac{1}{2}\mc{u_{1}\\v_1}D^{-1}\cdot\mc{-2u_2&1}+\frac{1}{2}\mc{1\\0}D^{-1}\cdot\mc{-2u_1u_2-2v_1&2u_1}
\end{multline*}
with
\[S_{11}=-u_1D^3-\frac32u_2D^2-\left(2u_1^3+2u_1v+\frac12u_3\right)D
-3u_1^2u_2-u_2v-u_1v_1.
\]
Two Hamiltonian operators are
\[
\bH_1=\begin{bmatrix}0 & u_{1}\\
-u_{1} & -D^{3}-2(u_{1}^{2}+v)D-2u_{1}u_{2}-v_{1}
\end{bmatrix}+\mc{1\\0}D^{-1}\mc{1&0},
\]
and
\[
\bH_2=\begin{bmatrix} -2 u_1 D-u_{2}& 
H_{12}
\\
H_{21} & H_{22}
\end{bmatrix}
+\mc{u_1\\v_1}D^{-1}\mc{v&v_t}+\mc{v\\v_t}D^{-1}\mc{u_1&v_1}
\]
with
\begin{gather*}
    \begin{split}
        H_{12}&=-2 D^4 -5 (u_1^2+v)D^2 -2 (5u_1 u_2+3 v_1)D 
\\
&\quad{}-u_1^4-v u_1^2-5 u_3 u_1-2 v^2-3 u_2^2-2 v_{2},
\\
    H_{21}&=-H_{12}-2v_1D-v_2,
\\
H_{22}&=2 u_{1} D^5 + 5 u_{2} D^4 + (4 u_{1}^3+2 v u_{1}+16 u_{3}) D^3
    \\
&\quad{}+ (18 u_{2} u_{1}^2+3 v_{1} u_{1}+3 v u_{2}+19 u_{4}) D^2
    \\
    &\quad{}+ (-4 v u_{1}^3+16 u_{3} u_{1}^2-4 v^2 u_{1}+8 u_{2}^2 u_{1}+v_{2} u_{1}
    -6 u_{2} v_{1}+7 v u_{3}+10 u_{5}) D
    \\
    &\quad{}-2 v_{1} u_{1}^3-6 v u_{2} u_{1}^2+5 u_{4} u_{1}^2-4 v v_{1} u_{1}
    +6 u_{2} u_3 u_{1}-2 u_{2}^3-2 v^2 u_{2}-v_{1} u_{3}
    \\
    &\quad{}-4 u_{2} v_{2}+3 v u_{4}+2 u_{6}.
    \end{split}
\end{gather*}
Using the first  Hamiltonian operator we can write the system in the Hamiltonian form: 
\[\begin{bmatrix}u\\
v
\end{bmatrix}_{t}=\mathbf{H}_{1}\frac{\delta}{\delta\mathbf{u}}\left(-u_{1}v-\frac{1}{3}u_{1}^{3}\right). 
\]
We cannot find any conserved density to express the system with the operator $\mathbf{H}_{2}$. However, possibly all higher symmetries of \eqref{eq:gu3} can be written in bi-Hamiltonian form. In particular, for the symmetry
\[
\mathbf{s}= \mc{v_{2}+\frac{1}{2}u_{2}^{2}+u_{1}^{2}v+\frac{1}{2}v^{2}+\frac{1}{6}u_{1}^{4}\\[4mm]
        \begin{gathered}\textstyle
            u_{6}+3u_{2}v_{2}+4u_{3}v_{1}+3vu_{4}+3u_{1}^{2}u_{4}+2u_{1}vv_{1}+\frac{2}{3}u_{1}^{3}v_{1}\\
    {}+12u_{1}u_{2}u_{3}+4u_{1}^{2}u_{2}v+2u_{2}v^{2}+2u_{1}^{4}u_{2}+4u_{2}^{3}
        \end{gathered}}
\]
we have 
\[
       \mathbf{s} =  \bH_2\frac{\delta\ }{\delta\mathbf{u}}\left(-\frac{v}2-\frac{u_1^2}{2}\right)
    =\bH_1\frac{\delta\ }{\delta\mathbf{u}}\left(u_{2} v_{1}+\frac{1}{2} u_{1} u_{2}^2-\frac{2}{3} vu_{1}^3-\frac{1}{2} v^2 u_{1}-\frac{u_{1}^5}{6}\right).
\]

\subsection{System~\eqref{eq:gu1}.}
In this case, writing the operators in quasi-local form requires symmetries and cosymmetries that explicitly depend on $x$ and $t$\footnote{When we expand the operators to pseudo-differential series this dependence disappears.}. Using the symmetries
\[
\mathbf{s}_1=\mc{1\\0},\ \mathbf{s}_2=\mc{x\\0},\ \mathbf{s}_3=\mc{t\\1},\ \mathbf{s}_4=\mc{xt\\x}
\]
and cosymmetries
\[\mathbf{c}_1=\mr{v_3\\-u_3},\ \mathbf{c}_2=\mc{2u_6+4u_2u_4+4u_3^2\\-2v_2},\ 
\mathbf{c}_3=\mc{-2 t u_{3}^2-2 v_2-2 t u_2 u_{4}-\frac{1}{2} x v_3-t u_{6}\\tv_{2}+\frac{1}{2}xu_{3}+u_{2}},
\]
the first two being variational derivatives of the conserved densities~
\[\rho_{1}=u_{2}v_{1},\qquad\rho_{2}=v_{1}^{2}-u_{3}^{2}+\frac{2}{3}u_{2}^{3},
\]
we find the following recursion operator:
\begin{multline*}
\R=\begin{bmatrix}
\frac{3}{4}v_{1} & D\\[2mm]
    D^{5}+\tfrac{5}{2}u_{2}D^{3}+\tfrac{5}{4}u_{3}D^{2}+\left(u_{4}+u_{2}^{2}\right)D
    -\tfrac{1}{2}u_{5}-u_{2}u_{3}
& \frac{3}{4}v_{1}
\end{bmatrix}
\\
+\frac{1}{4}\mc{x\\0}D^{-1}\cdot\mathbf{c}_1^++\frac{1}{4}\mc{t\\1}D^{-1}\cdot\mathbf{c}_2^+
+\frac{1}{2}\mc{1\\0}D^{-1}\cdot\mathbf{c}_3^+.
\end{multline*}
System~\eqref{eq:gu1} doesn't seem to possess any local higher symmetry. The recursion operator~$\R$ does not produce local higher symmetries from the classical symmetries. For example, appying~$\R$ to the $x$-shift symmetry~$(u_1,v_1)$ produces a nonlocal symmetry~$\mathbf{s}_{\text{nl}}$:
\begin{equation}\label{eq:nlsym}
\mathbf{s}_{\text{nl}}=\R\mc{u_1\\v_1}=\mc{v_2+D^{-1}(u_2v_1)\\
u_{6}+3u_{2}u_{4}+\frac{3}{2}u_{3}^{2}+\frac{1}{2}v_{1}^{2}+\frac{2}{3}u_{2}^{3}}.
\end{equation}
Two differential symplectic operators are
\[\bS_1=\mc{0 & -D^{2}\\
D^{2} & 0},\quad
\bS_2=\mc{-D^{7}-\frac{5}{2}u_{2}D^{5}-\frac{25}{4}u_{3}D^{4} & -\frac{3}{4}v_{1}D^{2}-v_{2}D-\frac{1}{4}v_{3}\\[0.5em]
\frac{3}{4}v_{1}D^{2}+\frac{1}{2}v_{2}D & D^{3}+\frac{1}{2}u_{2}D+\frac{1}{4}u_{3}
}.
\]
The recursion operator found is~$\R=\bS_1^{-1}\bS_2$.

A Hamiltonian operator is
\begin{multline}
\bH_1=\bS_1^{-1}=\mc{0&D^{-2}\\-D^{-2}&0}
\\
=
\begin{bmatrix}x\\0\end{bmatrix}D^{-1}\begin{bmatrix}t & 1\end{bmatrix}
+\begin{bmatrix}t\\1\end{bmatrix}D^{-1}\begin{bmatrix}x & 0\end{bmatrix}
-\begin{bmatrix}1\\0\end{bmatrix}D^{-1}\begin{bmatrix}xt & x\end{bmatrix}
-\begin{bmatrix}xt\\x\end{bmatrix}D^{-1}\begin{bmatrix}1 & 0\end{bmatrix}
\end{multline}
and the system can be written in Hamiltonian form as
\[\mc{u\\v}_t=\frac12\bH_1\frac{\delta\rho_2 }{\delta {\bf u}}.
\]
Another Hamiltonian operator is
\begin{multline*}
\bH_2=\R\bH_1
% =\mc{D^{-1}+\frac{1}{4}u_{2}D^{2}-\frac{1}{4}D^{-2}\cdot u_{2} & -\frac{1}{2}vD^{-1}+\frac{1}{2}D^{-1}\cdot v-\frac{1}{4}D^{-2}\cdot v_{1}\\[1em]
% \frac{1}{4}v_{1}D^{-2}+\frac{1}{2}D^{-1}\cdot v_{1}D^{-1} & 
% \begin{gathered}-D^{3}-\tfrac{5}{2}u_{2}D-\tfrac{5}{4}u_{3}
% \\
% {}-\tfrac{1}{2}\left(u_{4}+u_{2}^{2}\right)D^{-1}
% \\
% {}-\tfrac{1}{2}D^{-1}\cdot\left(u_{4}+u_{2}^{2}\right)
% \end{gathered}
% }
% \\
=\begin{bmatrix}0 & 0\\
0 & -D^{3}-\frac{5}{2}u_{2}D-\frac{5}{4}u_{3}
\end{bmatrix}-\frac{1}{4}\begin{bmatrix}x\\
0
\end{bmatrix}D^{-1}\begin{bmatrix}u_{1} & v_{1}\end{bmatrix}-\frac{1}{4}\begin{bmatrix}u_{1}\\
v_{1}
\end{bmatrix}D^{-1}\begin{bmatrix}x & 0\end{bmatrix}
\\
{}-\frac{1}{2}\begin{bmatrix}t\\
1
\end{bmatrix}D^{-1}\begin{bmatrix}v & v_{t}\end{bmatrix}
-\frac{1}{2}\begin{bmatrix}v\\
v_{t}
\end{bmatrix}D^{-1}\begin{bmatrix}t & 1\end{bmatrix}+\frac{1}{4}\begin{bmatrix}1\\
0
\end{bmatrix}D^{-1}\begin{bmatrix}xu_{1}+2tv & xv_{1}+2v+2tv_{t}\end{bmatrix}
\\
{}+\frac{1}{4}\begin{bmatrix}1+xu_{1}+2tv\\
xv_{1}+2v+2tv_{t}
\end{bmatrix}D^{-1}\begin{bmatrix}1 & 0\end{bmatrix}.
\end{multline*}
The given Hamiltonian operators are compatible since they are related by the shift $v\to v+\text{const}\,x$. Nevertheless, $\bH_2$ cannot be used to write a second Hamiltonian form of the system, in a local form. However, the non-local higher symmetry~$\mathbf{s}_{\text{nl}}$ (cf.~\eqref{eq:nlsym} \emph{can} be written in bi-Hamiltonian form as
\[\mathbf{s}_{\text{nl}}=\mathbf{H}_{1}\frac{\delta}{\delta\mathbf{u}}\left(\frac{1}{2}u_{4}^{2}-\frac{1}{2}v_{2}^{2}-\frac{3}{2}u_{2}u_{3}^{2}+\frac{1}{2}u_{2}v_{1}^{2}+\frac{1}{6}u_{2}^{4}\right)=\mathbf{H}_{2}
\frac{\delta\rho_{1}}{\delta\mathbf{u}}.
\]
Observe that the potentiation substitution~$u\to Du$, $v\to Dv$ transforms~\eqref{eq:gu2} into system~\eqref{eq:gu1}. This substitution transforms all local higher symmetries of~\eqref{eq:gu2} into non-local symmetries of\eqref{eq:gu1}, albeit preserving the local nature of conserved densities.

\section{Summary and concluding remarks}

We have explained a derivation of explicit integrability expressions for certain systems with a degenerate separant matrix, which we have defined as regularly diagonalisable system. These are systems that can be formally diagonalised after a diagonalisation procedure similar to the usual diagonalisation scheme \cite{MSY} for non-degenerate systems. 

As an application of the theory, using those explicit integrability conditions we provided a classification of two families of Boussinesq type systems admitting a finite number of integrability conditions. Even in these simplest cases some new examples arise. For the simplest family, we computed explicit expressions for recursion, symplectic and Hamiltonian operators, and gave a bi-Hamiltonian structure either for the system or for a higher symmetry (nonlocal in one case) proving the integrability of all the systems found. It is interesting that for writing some of these operators in quasi-local form, we had to resort to symmetries or cosymmetries explicitly dependent on the independent variables~$x$ and~$t$, although the systems do not have this dependence. Additionally, we had to use a cosymmetry which is not the variational derivative of a conserved density.

Many further computations can be done with the objects found in this work. A first development could be to perform a similar classification for the systems commented in Remarks~\ref{rem:sdx} and~\ref{rem:sysvux} and compare the result with those of the general polynomial case studied in~\cite{MNW,NW}. A study of the complete integrability of the systems given in~Theorem~\ref{thm:gux} would be also welcome. The question of why the general case with a complete~$g=g(u,u_1,u_2,v)$ (cf.~\eqref{eq:dc0}) is so difficult must be investigated. Most probably, the explanation is related with the existence of complicated differential substitutions relating different kinds (non-degenerate, regularly diagonalisable, etc.) of integrable systems
\[\begin{aligned}
    u_t&=f(u,u_1,\ldots,u_n;v,v_1,\ldots,v_m),\\
    v_t&=g(u,u_1,\ldots,u_n;v,v_1,\ldots,v_m).
\end{aligned}
\]
To systematically study these relations would constitute a whole fundamental project.

%===================================================================================
%===================================================================================
%===================================================================================
%===================================================================================

\section{Acknowledgements}

The authors want to thank V.S. Novikov for insightful comments and to A. Grishkov for his attention to the paper. RHH was supported by projects PID2021-124473NB-I00 and PID2019-106802GB-I00 from the Ministry of Science and Innovation of Spain, and is a member of the International Society of Nonlinear Mathematical Physics (ISNMP). VVS was supported by the FAPESR grant  2023/14060-7 (Brazil). He is grateful to IME USP for its hospitality during his visit to Sao Paulo. Both authors want to express their gratitude to the Institut des Hautes Études Scientifiques IHES (France)  for their hospitality and invaluable support.

\end{document}